\documentclass[lettersize,journal]{IEEEtran}
\usepackage{amsmath,amsfonts}
\usepackage{array}
\usepackage[caption=false,font=normalsize,labelfont=sf,textfont=sf]{subfig}
\usepackage{textcomp}
\usepackage{stfloats}
\usepackage{url}
\usepackage{verbatim}
\usepackage{graphicx}
\usepackage{threeparttable}
\usepackage{color}
\usepackage{amssymb}
\usepackage{algorithm}
\usepackage{algorithmic}
\floatname{algorithm}{Algorithm}
\usepackage{epstopdf}
\usepackage{cite}
\usepackage{amsthm}
\usepackage{booktabs}
\usepackage{multirow} 
\usepackage{lipsum}  
\usepackage[table]{xcolor}
\newtheorem{lemma}{Lemma}
\def\BibTeX{{\rm B\kern-.05em{\sc i\kern-.025em b}\kern-.08em
    T\kern-.1667em\lower.7ex\hbox{E}\kern-.125emX}}
\usepackage{balance}
\begin{document}
\title{Channel Estimation for RIS-Aided MU-MIMO mmWave Systems with Direct Channel Links}
\author{Taihao Zhang, Zhendong Peng, Cunhua Pan,~\IEEEmembership{Senior Member,~IEEE}, Hong Ren,~\IEEEmembership{Member,~IEEE}, and Jiangzhou Wang,~\IEEEmembership{Fellow,~IEEE}
	
	\thanks{Taihao Zhang, Cunhua Pan, Hong Ren and  Jiangzhou Wang are with National Mobile Communications Research Laboratory, Southeast University, Nanjing, China (e-mail:{220240854, cpan, hren, j.z.wang}@seu.edu.cn).
		 
		Corresponding author: Cunhua Pan.}
}


\maketitle

\begin{abstract}
In this paper, we propose a three-stage unified channel estimation strategy for reconfigurable intelligent surface (RIS)-aided multi-user (MU) multiple-input multiple-output (MIMO) millimeter wave (mmWave) systems with the existence of the direct channels, where the base station (BS), the users and the RIS are equipped with uniform planar array (UPA). The effectiveness of the developed three-stage strategy stems from the careful design of both the pilot signal sequence of the users and the vectors of RIS. Specifically, in Stage $\mathbf{\uppercase\expandafter{\romannumeral1}}$, the cascaded channel components are eliminated by configuring the RIS phase shift vectors with a $\pi$ difference to estimate the direct channels for all users. The orthogonal subspace projection is employed in Stage $\mathbf{\uppercase\expandafter{\romannumeral2}}$ to obtain equivalent signal matrices, enabling the estimation of angles of departure (AoDs) of the user-RIS channel for all users.
In Stage $\mathbf{\uppercase\expandafter{\romannumeral3}}$, we combine the signals of the time slots with the same pilots and project obtained measurement matrix to the orthogonal complement space of the component consisting of the portion of the direct channel, which removes the direct components and thus prevents error propagation from the direct channels to the cascaded channels. Then, we estimate the angles of arrival (AoAs) of the RIS-BS channel and remaining parameters of the cascaded channel for all users by exploiting the sparsity and correlation in the obtained equivalent matrices. Simulation results demonstrate that the proposed method yields better estimation performance than the existing methods.
\end{abstract}


\section{Introduction}

\IEEEPARstart{R}{econfigurable} intelligent surface (RIS) is perceived as a promising technology of 6G for improving the spectrum and energy efficiencies \cite{pan,RIS1,RIS2}. An RIS consists of a large number of cost-effective reflective elements \cite{multicell}.
It can effectively mitigate severe path loss and blockage issues to bring the tremendous benefits for millimeter wave (mmWave) multiple-input multiple-output (MIMO) systems. To reap the advantages of RIS, it is crucial to acquire accurate channel state information (CSI) \cite{CSI}. However, in RIS-aided MIMO systems, the users and BS are equipped with multiple antennas, and a large number of passive reflective elements in RIS lack the capacity of signal processing, which causes that numerous parameters need to be estimated \cite{fram}.

There have been contributions on the channel estimation in mmWave system, by exploiting the compressed sensing (CS) \cite{CS}, array processing \cite{aaray} and Bayesian frameworks \cite{bys} due to the fact that the channel is sparse and thus rank-deficient. Motivated by the works on the structured channel models \cite{CSI}, \cite{Peilan,dai,zhou,peng2022channel,peng2023two,anm} researched the channel estimation for RIS-aided mmWave systems. \cite{Peilan} exploited the sparsity to formulate the estimation for the user-RIS-BS cascaded channel as a sparse signal recovery problem. \cite{dai} employed the discrete Fourier transform (DFT) method to estimate the angle of arrival (AoA) for the multi-user (MU) system. A significant reduction in pilot overhead was achieved for MU mmWave systems in \cite{zhou,peng2022channel}.

Nevertheless, the aforementioned channel estimation all assumed that the direct channels are blocked. However, in practice, the cascaded and direct channel may coexist, leading to the respective components coupled in the received signal \cite{zhouguixiu}. It is challenging to achieve simultaneous estimation of the cascaded and direct channel. There are a few studies that considered channel estimation for RIS-aided mmWave systems in the presence of direct channels, which can be categorized into two types based on the behavior of the RIS: ON-OFF method \cite{ONOFF,ONOFF2,ON-OFF3} and Always-ON method \cite{alwayson,taihao,alwayson2,alwayson3}. In the ON-OFF framework, they first turned off the RIS to estimate the direct channels. After reactivating the RIS, \cite{ONOFF2} subtracted the estimated direct component and then proposed two novel algorithms based on the properties of the channel matrix to estimate the cascaded channel, while \cite{ON-OFF3} utilized deep learning-based method to first estimate the composite channel and then estimate the cascaded channel by subtracting the direct channel estimate from the composite channel estimate.
There exist two severe issues for ON-OFF method. First, it is impractical to frequently switch on-off the RIS in each time slot \cite{Phys}, which causes reflection power loss and additional implementation cost. Second, it suffers from error propagation for the estimation error of the direct channel to the cascaded channel estimation, which is called direct-to-cascaded error propagation in this paper for convenience, leading to the degradation in estimation performance. The Always-ON framework avoids frequent turning on/off RIS enabling the simultaneous estimation of the direct and cascaded channel. The methods proposed in \cite{alwayson,taihao,alwayson2} were designed for RIS-aided multiple-input single-output (MISO) systems. \cite{alwayson} avoided switching the RIS by designing phase shifts, but this amplified the noise in the equivalent cascaded signal. To address this issue, \cite{taihao} introduced a method based on projection combined with carefully designed RIS. In \cite{alwayson2}, based on the assumption of a known prior signal distribution, the maximum likelihood estimation (MLE) was used to represent the direct channel estimate, which was then treated as known in the overall estimation problem. \cite{alwayson3} was tailored for multi-user systems. It first suppressed the direct channels to estimate AoAs of RIS-BS channel. This estimated AoA was then used as input to estimate the parameters of the direct channels and the remaining parameters, which achieved joint estimation but also introduced error propagation among the estimated parameters.
These studies have made significant contributions.

Nevertheless, there is a paucity of studies on the channel estimation framework for RIS-aided MU-MIMO systems in the presence of direct channels. The complexity of the signal models and the enormous parameters to be estimated render the simultaneous estimation of the direct and cascaded channels an extremely challenging task.
Against the background, we propose a novel three-stage unified channel estimation framework for RIS-aided MU-MIMO mmWave systems. The major contributions of the paper are summarized as follows.
\begin{itemize}
	\item {We propose a three-stage unified channel estimation framework for the considered RIS-aided MU-MIMO systems. The pilot signal sequence of the users and the vectors of RIS are meticulously designed to achieve separation of the direct and cascaded channels. 
	The developed method has three significant advantages. First, it does not require switching on-off the RIS, thereby reducing the power loss caused by frequent switching. Second, it halves the noise power for the direct channel estimation. Third, it eliminates the direct-to-cascaded error propagation effect and keeps the noise power unchanged for the cascaded channel estimation process.
	}
	
	\item{The proposed channel estimation protocol is carefully designed to achieve two objectives: $(i)$ decouple the direct and cascaded channel so as to avoid direct-to-cascaded error propagation, and $(ii)$ maximize the utilization of elements in the support set required for sparse signal estimation, thereby reducing pilot overhead.
	}

	\item {In Stage $\mathrm{\uppercase\expandafter{\romannumeral1}}$, the components of the cascaded channels are eliminated by combining signals from two time slots with the RIS phase shifts differing by $\pi$ radians. In Stage $\mathrm{\uppercase\expandafter{\romannumeral2}}$, to eliminate the components of the direct channels, signals with identical pilot vectors are stacked and then the orthogonal subspace projection is applied. In Stage $\mathrm{\uppercase\expandafter{\romannumeral3}}$, according to the designed pilots and the vectors of RIS, we first combine the time slot signals with identical pilots, and project them onto the orthogonal complement space of the direct channel components, yielding the equivalent cascaded signals.}

	\item{In order to reduce the pilot overhead and complexity in multi-user scenarios, this work features two key innovations. First, leveraging the commonality of across different users, we combine the time slot signals with the identical equivalent vectors of RIS to obtain the measurement matrix. The UPA-DFT method\footnote{It is defined as an angle estimation technique that exploits the properties of the UPA-type DFT matrix codebook combined with the Kronecker product.} is used to estimate the common AoAs of RIS-BS channel. Second, using the idea of equivalent virtual single antenna conversion (EVSA), which transforms the estimation of MIMO mmWave channels into multiple equivalent estimation of MISO mmWave channels, we employ vectorized OMP to estimate the remaining parameters.}
		 
\end{itemize}

The remainder of this paper is organized as follows. The system model is introduced in Section \ref{systemmodel}. Section \ref{section3} presents the details of the proposed strategy and analyzes the pilot overhead. The simulation results are given in Section \ref{sim}. Finally, Section \ref{conl} provides the conclusions.

\emph{Notations}: Vectors and matrices are denoted by boldface lowercase letters and boldface uppercase letters. The symbols $\mathbf{X}^*$, $\mathbf{X}^{\mathrm{T}}$, $\mathbf{X}^{\mathrm{H}}$ and $\mathbf{X}^{\dagger}$ represent the conjugate, transpose, Hermitian and pseudo-inverse of matrix $\mathbf{X}$. $\left\| \mathbf{X} \right\| _F$ and $\left\| \mathbf{x} \right\| _2$ denote the Frobenius norm of matrix $\mathbf{X}$ and the Euclidean norm of vector $\mathbf{x}$. $\mathrm{Diag}\left\{ \mathbf{x} \right\}$ is a diagonal matrix with the entries of vector $\mathbf{x}$ on its diagonal, while $\mathrm{diag}\left\{ \mathbf{X} \right\}$ represent a vector whose entries are extracted from the diagonal entries of the matrix $\mathbf{X}$. $\mathrm{vec}\left( \mathbf{X} \right)$ denotes the vectorization of $\mathbf{X}$ obtained by stacking the columns of matrix $\mathbf{X}$. $\left[ \mathbf{x} \right] _m$ is the $m$-th element of the vector $\mathbf{x}$, and $\left[ \mathbf{X} \right] _{m,n}$ is the $(m,n)$-th element of the matrix $\mathbf{X}$. The $m$-th row and the $n$-th column of matrix $\mathbf{X}$ are denoted as $\mathbf{X}_{\left( m,: \right)}$ and $\mathbf{X}_{\left( :,n \right)}$, respectively.
The Kronecker product and Khatri-Rao product between two matrices $\mathbf{X}$ and $\mathbf{Y}$ are denoted by $\mathbf{X}\otimes \mathbf{Y}$ and $\mathbf{X}\diamond \mathbf{Y}$. Moreover, $\mathbf{I}_N$ is the $N\times N$ identity matrix. $\mathbf{1}_N$ is the all-one vector of size $N$. $\mathbb{E} \left\{ \cdot \right\}$ is the expectation operation. $\lceil x \rceil$ rounds up to the nearest integer. $\left| \cdot \right|$ represents the modulus of a complex number. For a linear space $\boldsymbol{U}$, $\boldsymbol{U}^{\bot}$ is the orthogonal complement of $\boldsymbol{U}$ and $\mathrm{dim}\left( \boldsymbol{U} \right)$ represent the dimension of $\boldsymbol{U}$.
\section{System Model}\label{systemmodel}

This paper investigates uplink channel for the narrow-band time-division duplex (TDD) RIS-aided multi-user mmWave system, in which $K$ multi-antenna users establish a communication link with a BS equipped with an $N=N_1\times N_2$ antenna UPA, where $N_1$ and $N_2$ are the numbers of vertical and horizontal antennas, respectively. Taking user $k$ $(k\in \mathcal{K} =\left\{ 1,\ldots,K \right\})$ as an example, the user $k$ is equipped with $Q_k=Q_{k_1}\times Q_{k_2}$ antenna UPA ($Q_{k_1}$ vertical and $Q_{k_2}$ horizontal antennas). An RIS with an $M=M_1\times M_2$ passive reflecting elements UPA ($M_1$ vertical and $M_2$ horizontal elements) is deployed to improve communication. We consider quasi-static channels, where each channel remains approximately unchanged with an estimation period.
\begin{figure}
	\centering
	\includegraphics[width=8.5cm]{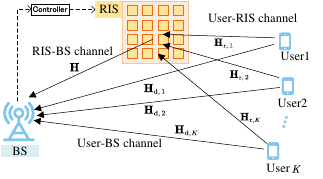}
	\caption{An RIS-aided MU-MIMO mmWave communication system.}
	\label{case} 
\end{figure}

The far-field scenario is considered, i.e., the Rayleigh distance of the MIMO/RIS system is significantly less than the communication distance.
$\mathbf{H}_{\mathrm{d},k}\in \mathbb{C} ^{N\times Q_k}$, $\mathbf{H}_{\mathrm{r},k}\in \mathbb{C} ^{M\times Q_k}$ and  $\mathbf{H}\in \mathbb{C} ^{N\times M}$ is denoted as the channel matrix from user $k$ to the BS, from user $k$ to the RIS, and from the RIS and the BS, respectively. Exploiting the geometric channel model \cite{SV2} typically used for representing the sparsity of mmWave channel, the channel matrices are modeled as follows,
\begin{align}\label{h_d}
\mathbf{H}_{\mathrm{d},k}&=\sum_{i=1}^{I_k}{\zeta _{k,i}\mathbf{a}_N\left( \sigma _{k,i},\xi _{k,i} \right) \mathbf{a}_{Q_k}^{\mathrm{H}}\left( \rho _{k,i},\kappa _{k,i} \right)}\nonumber
\\
&=\mathbf{D}_{N,k}\mathbf{Z}_k\mathbf{D}_{Q,k}^{\mathrm{H}}\,\,  ,\forall k\in \mathcal{K} ,
\end{align}
\begin{align}\label{h_r}
\mathbf{H}_{\mathrm{r},k}&=\sum_{j=1}^{J_k}{\beta _{k,j}\mathbf{a}_M\left( \varphi _{k,j},\theta _{k,j} \right) \mathbf{a}_{Q_k}^{\mathrm{H}}\left( \eta _{k,j},\chi _{k,j} \right)}\nonumber
\\
&=\mathbf{A}_{M,k}\mathbf{B}_k\mathbf{A}_{Q,k}^{\mathrm{H}}\,\, ,\forall k\in \mathcal{K},
\end{align}
\begin{align}\label{H}
\mathbf{H}=\sum_{l=1}^L{\alpha _l\mathbf{a}_N\left( \psi _l,\nu _l \right) \mathbf{a}_{M}^{\mathrm{H}}\left( \omega _l,\mu _l \right)}=\mathbf{A}_N\mathbf{\Lambda A}_{M}^{\mathrm{H}},
\end{align}
where $I_k$, $J_k$ and $L$ are the numbers of propagation paths from user $k$ to the BS, from user $k$ to the RIS and from the RIS to the BS. $\zeta _{k,i}$, $\left( \sigma _{k,i},\xi _{k,i} \right)$ and $\left( \rho _{k,i},\kappa _{k,i} \right)$ denote the complex gain, AoA and AoD of the $i$-th path in the user $k$-BS channel. $\beta _{k,j}$, $\left( \varphi _{k,j},\theta _{k,j} \right)$ and $\left( \eta _{k,j},\chi _{k,j} \right)$ are the complex gain, AoA and AoD of the $j$-th path in the user $k$-RIS channel. $\alpha _l$, $\left( \psi _l,\nu _l \right)$ and $\left( \omega _l,\mu _l \right)$ represent the complex gain, AoA and AoD of the $l$-th path in the RIS-BS channel. Moreover, $\mathbf{D}_{N,k}=\left[ \mathbf{a}_N\left( \sigma _{k,1},\xi _{k,1} \right),\ldots,\mathbf{a}_N\left( \sigma _{k,I_k},\xi _{k,I_k} \right) \right] \in \mathbb{C} ^{N\times I_k}$, $\mathbf{D}_{Q,k}=\left[ \mathbf{a}_{Q_k}\left( \rho _{k,1},\kappa _{k,1} \right),\ldots,\mathbf{a}_{Q_k}\left( \rho _{k,I_k},\kappa _{k,I_k} \right) \right] \in \mathbb{C} ^{Q_k\times I_k}$ and $\mathbf{Z}_k=\mathrm{Diag}\left\{ \zeta _{k,1},\ldots,\zeta _{k,I_k} \right\} \in \mathbb{C} ^{I_k\times I_k}$ are the AoA steering matrix, AoD steering matrix and complex gain matrix of the user $k$-BS channel. $\mathbf{A}_{M,k}=\left[ \mathbf{a}_M\left( \varphi _{k,1},\theta _{k,1} \right) ,\ldots,\mathbf{a}_M\left( \varphi _{k,J_k}\theta _{k,J_k} \right) \right] \in \mathbb{C} ^{M\times J_k}$, $\mathbf{A}_{Q,k}=\left[ \mathbf{a}_{Q_k}\left( \eta _{k,1},\chi _{k,1} \right) ,\ldots,\mathbf{a}_{Q_k}\left( \eta _{k,J_k},\chi _{k,J_k} \right) \right] \in \mathbb{C} ^{Q_k\times J_k}$ and $\mathbf{B}_k=\mathrm{Diag}\left\{ \beta _{k,1},\ldots,\beta _{k,J_k} \right\} \in \mathbb{C} ^{J_k\times J_k}$ denote the AoA steering matrix, AoD steering matrix and complex gain matrix of the user $k$-RIS channel. $\mathbf{A}_N=\left[ \mathbf{a}_N\left( \psi _1,\nu _1 \right) ,\ldots,\mathbf{a}_N\left( \psi _L,\nu _L \right) \right] \in \mathbb{C} ^{N\times L}$, $\mathbf{A}_M=\left[ \mathbf{a}_M\left( \omega _1,\mu _1 \right) ,\ldots,\mathbf{a}_M\left( \omega _L,\mu _L \right) \right] \in \mathbb{C} ^{M\times L}$ and $\mathbf{\Lambda }=\mathrm{Diag}\left\{ \alpha _1,\ldots,\alpha _L \right\} \in \mathbb{C} ^{L\times L}$ are AoA steering matrix, AoD steering matrix and complex gain matrix of RIS-BS channel.

Here, the aforementioned expressions in the form of $\mathbf{a}_F\left( x,y \right) \in \mathbb{C} ^{F\times 1}$ represent the array steering vectors of a UPA with dimension $F_1 \times F_2$ ($F_1$ and $F_2$ denote the numbers of elements in vertical and horizontal direction, respectively), which can be expressed as
$\mathbf{a}_F\left( x,y \right) =\mathbf{a}_{F_1}\left( x \right) \otimes \mathbf{a}_{F_2}\left( y \right)$,
where $\mathbf{a}_{F_1}\left( x \right)$ and $\mathbf{a}_{F_2}\left( y \right)$ are the steering vectors with respect to the vertical and horizontal direction of the UPA, respectively, and have the form as $\mathbf{a}_P\left( u \right) =[ e^{-\mathrm{i}2\pi \mathbf{p}u}] ^{\mathrm{T}}, \mathbf{p}=\left[ 0,\ldots,P-1 \right]$. Furthermore, we denote $x$ and $y$ as equivalent spatial frequencies of the $x$-axis and $y$-axis of the UPA given by $x=\frac{d}{\lambda}\cos \left( \epsilon \right), y=\frac{d}{\lambda}\sin \left( \epsilon \right) \cos \left( \iota  \right)$,
where $\epsilon \in \left[ -\pi /2,\pi /2 \right)$ and $\iota  \in \left[ -\pi ,\pi \right)$ are the signal elevation and azimuth angles of the UPA, respectively. $d$ is the element spacing and $\lambda$ is the carrier wavelength. The collected array adopts half wavelength antenna spacing at the user, BS and RIS, in order to avoid severe side lobes, form a single pair main lobes without phase ambiguity \cite{chenxian}, and mitigate mutual coupling effect of RIS \cite{RISouhe}.

Furthermore, denote $\mathbf{s}_{k,t} \in \mathbb{C} ^{Q_k\times 1}$ and $\mathbf{e}_t\in \mathbb{C} ^{M\times 1}$ as the pilot signal of user $k$ and the vector of the RIS in time slot $t$. The users take the one-by-one transmission scheme\footnote{The one-by-one transmission scheme is equivalent to the orthogonal pilot transmission scheme in terms of pilot overhead \cite{PILOT}. The proposed channel estimation protocol remains applicable to the orthogonal pilot scheme without incurring additional pilot overhead.} to avoid inter-user interference.
	Normalizing transmit power, the received signal from user $k$ in the time slot $t$ is given by
\begin{align}\label{yk}
&\mathbf{y}_{k,t} =\mathbf{H}_{\mathrm{d},k}\mathbf{s}_{k,t} +\mathbf{H}\mathrm{Diag}\left\{ \mathbf{e}_t \right\} \mathbf{H}_{\mathrm{r},k}\mathbf{s}_{k,t} +\mathbf{n}_{k,t}  \nonumber
\\
&\overset{(a)}{=}\mathbf{H}_{\mathrm{d},k}\mathbf{s}_{k,t} +\left( \mathbf{s}_{k,t}^{\mathrm{T}} \otimes \mathbf{I}_N \right) \mathrm{vec}\left( \mathbf{H}\mathrm{Diag}\left\{ \mathbf{e}_t \right\} \mathbf{H}_{\mathrm{r},k} \right)+\mathbf{n}_{k,t}   \nonumber
\\
&\overset{(b)}{=}\mathbf{H}_{\mathrm{d},k}\mathbf{s}_{k,t} +\left( \mathbf{s}_{k,t}^{\mathrm{T}} \otimes \mathbf{I}_N \right) \left( \mathbf{H}_{\mathrm{r},k}^{\mathrm{T}}\diamond \mathbf{H} \right) \mathbf{e}_t+\mathbf{n}_{k,t} \nonumber
\\
&\triangleq\mathbf{H}_{\mathrm{d},k}\mathbf{s}_{k,t} +\left( \mathbf{s}_{k,t}^{\mathrm{T}} \otimes \mathbf{I}_N \right) \mathbf{G}_{k}\mathbf{e}_t+\mathbf{n}_{k,t}  ,
\end{align}
where $\mathbf{n}_{k,t} \sim \mathcal{C} \mathcal{N} \left( 0,\delta ^2\mathbf{I} \right)$ is additive white Gaussian noise (AWGN) in time slot $t$ during user $k$ transmission at the BS. 
Here, ($a$) vectorizes $\mathbf{y}_{k,t}$ and uses $\mathrm{vec}\left( \mathbf{ABC} \right) =\left( \mathbf{C}^{\mathrm{T}}\otimes \mathbf{A} \right) \mathrm{vec}\left( \mathbf{B} \right)$, and ($b$) uses $\mathrm{vec}\left( \mathbf{A}\mathrm{Diag}\left\{ \mathbf{b} \right\} \mathbf{C} \right) =\left( \mathbf{C}^{\mathrm{T}}\diamond \mathbf{A} \right) \mathbf{b}$ \cite{Zhang_2017}.
$\mathbf{G}_{k}\triangleq \mathbf{H}_{\mathrm{r},k}^{\mathrm{T}}\diamond \mathbf{H}\in \mathbb{C} ^{Q_kN\times M}$ is regarded as the cascaded channel for transmission design in the RIS-aided MIMO systems \cite{peng2022channel}, which can be expressed as
\begin{align}\label{G}
\mathbf{G}_{k}&=\mathbf{H}_{\mathrm{r},k}^{\mathrm{T}}\diamond \mathbf{H}
=\left( \mathbf{A}_{M,k}\mathbf{B}_k\mathbf{A}_{Q,k}^{\mathrm{H}} \right) ^{\mathrm{T}}\diamond \left( \mathbf{A}_N\mathbf{\Lambda A}_{M}^{\mathrm{H}} \right) \nonumber
\\
&=\left( \mathbf{A}_{Q,k}^{*}\otimes \mathbf{A}_N \right) \left( \mathbf{B}_k\otimes \mathbf{\Lambda } \right) \left( \mathbf{A}_{M,k}^{\mathrm{T}}\diamond \mathbf{A}_{M}^{\mathrm{H}} \right) ,
\end{align}
where $\left( \mathbf{AC} \right) \diamond \left( \mathbf{BD} \right) =\left( \mathbf{A}\otimes \mathbf{B} \right) \left( \mathbf{C}\diamond \mathbf{D} \right)$ and $\left( \mathbf{AB} \right) \otimes \left( \mathbf{CD} \right) =\left( \mathbf{A}\otimes \mathbf{C} \right) \left( \mathbf{B}\otimes \mathbf{D} \right)$ are used. We develop a novel channel estimation protocol for estimating direct channels $\mathbf{H}_{\mathrm{d},k}$ and cascaded channels $\mathbf{G}_{k}$ in the RIS-aided MU-MIMO mmWave system, which can avoid switching on-off RIS and eliminate direct-to-cascaded error propagation completely.

\section{Unified Channel Estimation For a RIS-aided MU-MIMO System}\label{section3}
\subsection{Channel Estimation Protocol}
\begin{figure*} 
	\centering
	\includegraphics[width=17cm]{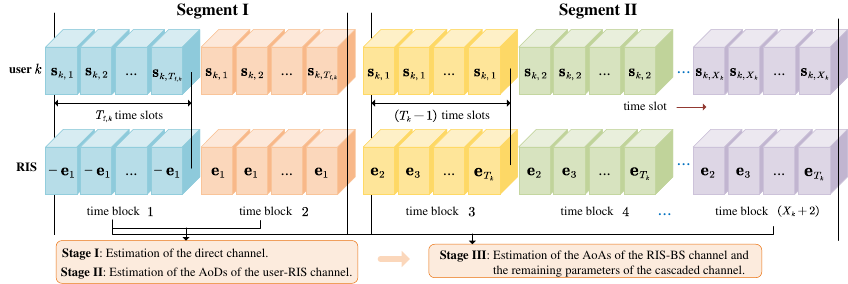}
	\caption{The sequence diagram of the proposed unified channel estimation protocol for the RIS-aided MIMO mmWave system.}
	\label{fig1} 
\end{figure*}
Taking user $k$ as an example, the sequence diagram of the proposed channel estimation protocol is presented in Fig. \ref{fig1}, where the first row represents the pilot signals transmitted by user $k$ and the second row denotes the design of the vectors of RIS. The protocol is divided into two segments along the time dimension: Segment $\mathrm{\uppercase\expandafter{\romannumeral1}}$ corresponds to Stages $\mathrm{\uppercase\expandafter{\romannumeral1}}$ and $\mathrm{\uppercase\expandafter{\romannumeral2}}$, while Segment $\mathrm{\uppercase\expandafter{\romannumeral2}}$ aligns with Stage $\mathrm{\uppercase\expandafter{\romannumeral3}}$.
We deliberately design the pilot transmission schemes, RIS phase shifts, and signal processing methods to  achieve decoupling the direct and cascaded channel components in received signals, while also obtaining adequate independent measurements \cite{omppiot} for enhanced parameter estimation accuracy.

In Segment $\mathrm{\uppercase\expandafter{\romannumeral1}}$, the first two time blocks are designed, each of which includes $T_{\mathrm{f},k}$ time slots. The vectors of RIS within time block $1$ are designed as the opposite value\footnote{The phase shift of each RIS is controlled via a finite-bit quantized structure (typically 1, 2, 3, 4-bit in practice) \cite{cui2014coding}. Under the $N$-bit structure, the opposite RIS values can be achieved by precisely setting element $a$ to $p_a=n\pi /\left( 2^{N-1} \right) , n=0,1,...,2^N-1$ and element $b$ to $p_b=\left( p_a+\pi \right) \mathrm{mod}2\pi$.} to those within time block $2$, while the pilot signals remain identical across both time blocks. Based on the designed RIS phase shift vectors and pilot signals, in Stage $\mathrm{\uppercase\expandafter{\romannumeral1}}$, the equivalent signals is obtained to estimate the direct channels.
In Stage $\mathrm{\uppercase\expandafter{\romannumeral2}}$, with the same time block in this segment, we recombine the received signals and utilize the orthogonal complement space projection to achieve the estimation of the AoDs of the user-RIS channel.

On the other hand, Segment $\mathrm{\uppercase\expandafter{\romannumeral2}}$ consists of $X_k$ time blocks, i.e., from time block $3$ to time block $(X_k+2)$, and each time block consists of $(T_k-1)$ time slots. 
The value of the pilot signals within each time block in Segment $\mathrm{\uppercase\expandafter{\romannumeral2}}$ are designed to be the same as those within the corresponding time slots of the preceding two time blocks. For example, within time block $(x+2)$, the value of pilot symbols is set to the value of pilots in time slot $x$ within time block 1.
Additionally, in Segment $\mathrm{\uppercase\expandafter{\romannumeral2}}$, the vectors of RIS within all $X_k$ time blocks are identically designed as $\left\{\mathbf{e}_2,\mathbf{e}_3,\ldots,\mathbf{e}_{T_k} \right\}$. In Stage $\mathrm{\uppercase\expandafter{\romannumeral3}}$, the AoAs of the RIS-BS channel and the remaining parameters of the cascaded channels can be estimated by utilizing a total of $(X_k+2)$ time blocks from Segment $\mathrm{\uppercase\expandafter{\romannumeral1}}$ and Segment $\mathrm{\uppercase\expandafter{\romannumeral2}}$.
Specifically, by re-organizing the collected $(X_k+2)$ time blocks, a novel signal pre-processing technique known as the orthogonal complement space projection is proposed to obtain the equivalent signal for cascaded channel estimation. It can be observed that the proposed method completely prevents the direct-to-cascaded channel estimation error propagation and thus enhances the estimation performance of the cascaded channel estimation.
For the sake of illustration, in the remaining part of the paper, we denote $\mathbf{n}_{k,t}^{u}$ as the noise vector at the BS in time slot $t$ of time block $u$ during user $k$ transmission.

\subsection{Stage ${\uppercase\expandafter{\romannumeral1}}$: Estimation of the direct channel}\label{stage1}
In Stage $\mathrm{\uppercase\expandafter{\romannumeral1}}$, we utilize the time blocks of Segment $\mathrm{\uppercase\expandafter{\romannumeral1}}$ to estimate all the parameters of the direct channel. We denote the pilot overhead required for estimating the direct channel by $T_{\mathrm{d},k}$, $\left( T_{\mathrm{d},k} \leqslant T_{\mathrm{f},k} \right)$.
The received signal processing procedure for Stage $\mathrm{\uppercase\expandafter{\romannumeral1}}$ is shown in the upper half of Fig. \ref{figfirst}.
\begin{figure}
	\centering
	\includegraphics[width=8.5cm]{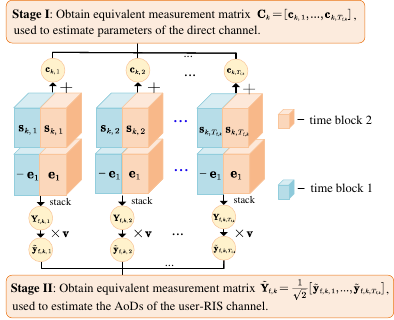}
	\caption{The schematic diagram of signal pre-processing in Stage $\mathrm{\uppercase\expandafter{\romannumeral1}}$ and $\mathrm{\uppercase\expandafter{\romannumeral2}}$.}
	\label{figfirst} 
\end{figure}

Based on \eqref{yk}, we can obtain the received signal at the BS in time slot $t\left( 1\leqslant t\leqslant T_{\mathrm{d},k} \right)$ within time block $1$ and $2$ as
\begin{align}\label{y1k}
\mathbf{y}_{k,t}^{1}=\mathbf{H}_{\mathrm{d},k}\mathbf{s}_{k,t}+\left( \mathbf{s}_{k,t}^{\mathrm{T}}\otimes \mathbf{I}_N \right) \mathbf{G}_{k}\left( -\mathbf{e}_1 \right) +\mathbf{n}_{k,t}^{1} ,
\end{align}
\begin{align}\label{y2k}
\mathbf{y}_{k,t}^{2}=\mathbf{H}_{\mathrm{d},k}\mathbf{s}_{k,t}+\left( \mathbf{s}_{k,t}^{\mathrm{T}}\otimes \mathbf{I}_N \right) \mathbf{G}_{k}\mathbf{e}_1+\mathbf{n}_{k,t}^{2} ,
\end{align}
where $\mathbf{n}_{k,t}^{1} $ and $\mathbf{n}_{k,t}^{2} $ represent the corresponding noise vectors within time block $1$ and $2$, and $\mathbf{s}_{k,t}$ denotes the pilot signals in time slot $t$ within both time block $1$ and time block $2$.  
Identical pilots with the opposite vectors of RIS enable the obtaining of equivalent direct channel signals $\mathbf{c}_{k,t}$ given by,
\begin{align}\label{ct}
\mathbf{c}_{k,t}=\frac{\mathbf{y}_{k,t}^{1}+\mathbf{y}_{k,t}^{2}}{2}=\mathbf{D}_{N,k}\mathbf{Z}_k\mathbf{D}_{Q,k}^{\mathrm{H}}\mathbf{s}_{k,t}+\frac{\mathbf{n}_{k,t}^{1}+\mathbf{n}_{k,t}^{2}}{2}.
\end{align}
The design of the time-varying pilots across slots is engineered to generate independent measurements.
Based on this design, performing the similar steps to \eqref{ct} for all $T_{\mathrm{d},k}$-pairs time slots to obtain the overall $T_{\mathrm{d},k}$ equivalent received signals and stacking them as the equivalent measurement matrix $\mathbf{C}_k=[ \mathbf{c}_{k,1},\ldots,\mathbf{c}_{k,T_{\mathrm{d},k}} ] \in \mathbb{C} ^{N\times T_{\mathrm{d},k}}$, we have
\begin{align}\label{CkH}
\mathbf{C}_{k}^{\mathrm{H}}=\mathbf{S}_{k}^{\mathrm{H}}\mathbf{D}_{Q,k}\mathbf{\Theta }_k+\mathbf{N}_{k}^{\mathrm{H}},
\end{align}
where $\mathbf{\Theta }_k\triangleq \mathbf{Z}_{k}^{\mathrm{H}}\mathbf{D}_{N,k}^{\mathrm{H}}$, $\mathbf{S}_k=[ \mathbf{s}_{k,1},\ldots,\mathbf{s}_{k,T_{\mathrm{d},k}} ] \in \mathbb{C} ^{Q_k\times T_{\mathrm{d},k}}$ is the stacked pilot matrix consisting of $T_{\mathrm{d},k}$ pilot vectors and $\mathbf{N}_k=\left[ \frac{1}{2}\left( \mathbf{n}_{k,t}^{1}+\mathbf{n}_{k,t}^{2} \right) ,\ldots,\frac{1}{2}\left( \mathbf{n}_{k,T_{\mathrm{d},k}}^{1}+\mathbf{n}_{k,T_{\mathrm{d},k}}^{2} \right) \right] \in \mathbb{C} ^{N\times T_{\mathrm{d},k}}$ represents the corresponding noise matrix. To estimate $\mathbf{D}_{Q,k}$, \eqref{CkH} can be re-expressed by using the virtual angular domain (VAD) representation as
\begin{align}\label{CkHVAD}
\mathbf{C}_{k}^{\mathrm{H}}=\mathbf{S}_{k}^{\mathrm{H}}\left( \mathbf{D}_1\otimes \mathbf{D}_2 \right) \tilde{\mathbf{\Theta}}_k+\mathbf{N}_{k}^{\mathrm{H}},
\end{align}
where $\mathbf{D}_1\in \mathbb{C} ^{Q_{k_1}\times B_1}$ and $\mathbf{D}_2\in \mathbb{C} ^{Q_{k_2}\times B_2}$ are overcomplete dictionary matrices ($B_1\geqslant Q_{k_1}, B_2\geqslant Q_{k_2}$). The range of atoms in $\mathbf{D}_1$ and $\mathbf{D}_2$ are $[ -\frac{d}{\lambda _c},( 1-\frac{2}{B_1} ) \frac{d}{\lambda _c} ]$ and $[ -\frac{d}{\lambda _c},( 1-\frac{2}{B_2} ) \frac{d}{\lambda _c} ]$ with the resolution of $2/B_1$ and $2/B_2$, respectively. $\tilde{\mathbf{\Theta}}_k\in \mathbb{C} ^{B_1B_2\times N}$ is a row-sparse matrix with $I_k$ non-zero rows. \eqref{CkHVAD} can be regarded as a sparse signal recovery problem under multiple measurement vector (MMV) case. Employing the simultaneous orthogonal matching pursuit (SOMP) \cite{zd1} algorithm , the AoDs of the user-BS channel can be estimated, i.e., $\left( \rho _{k,i},\kappa _{k,i} \right)$ for $\forall i$. Then, $\rho _{k,i}$ and $\kappa _{k,i}$ are obtained by using the properties of the Kronecker product. Specifically, assuming that the $b$-th row of the matrix $\tilde{\mathbf{\Theta}}_k$ is nonzero, and the $b$-th column of $( \mathbf{D}_1\otimes \mathbf{D}_2 )$ is the steering vector associated with an AoD pair. The indices of elevation and azimuth angles in $ \mathbf{D}_1$ and $ \mathbf{D}_2$ are obtained by
\begin{align}\label{b1b2}
		b_1=\left\lceil {\frac{b}{B_2}} \right\rceil , b_2=b-B_2\left( b_1-1 \right).
\end{align}
With the obtained $\hat{\mathbf{D}}_{Q,k}$, we estimate the remaining parameters of the user-BS channel. Processing the measurement matrix $\mathbf{C}_{k}$ in \eqref{CkH} by right-multiplying $( \hat{\mathbf{D}}_{Q,k}^{\mathrm{H}}\mathbf{S}_k ) ^{\dagger}$, we have
\begin{align}\label{persudo}
\tilde{\mathbf{C}}_k&\triangleq ( \mathbf{D}_{N,k}\mathbf{Z}_k( \hat{\mathbf{D}}_{Q,k}+\Delta \mathbf{D}_{Q,k} ) ^{\mathrm{H}}\mathbf{S}_k+\mathbf{N}_k ) ( \hat{\mathbf{D}}_{Q,k}^{\mathrm{H}}\mathbf{S}_k ) ^{\dagger}\nonumber
\\
&=\mathbf{D}_{N,k}\mathbf{Z}_k+\tilde{\mathbf{N}}_k,
\end{align}
where $\Delta \mathbf{D}_{Q,k}$ represents the estimation error of $\mathbf{D}_{Q,k}$ and $\tilde{\mathbf{N}}_k=(\mathbf{D}_{N,k}\mathbf{Z}_k \Delta \mathbf{D}_{Q,k}^{\mathrm{H}}\mathbf{S}_k+\mathbf{N}_k ) ( \hat{\mathbf{D}}_{Q,k}^{\mathrm{H}}\mathbf{S}_k ) ^{\dagger}$ stands for the equivalent noise matrix. To decouple the remaining parameters, we can obtain the $i$-th column of $\tilde{\mathbf{C}}_k$, which is associated the $i$-th path of user-BS channel, given by
\begin{align}\label{csnlls}
\mathbf{c}_{k,i}=[ \tilde{\mathbf{C}}_k ] _{\left( :,i \right)}=\zeta _{k,i}\mathbf{a}_N\left( \sigma _{k,i},\xi _{k,i} \right) +\tilde{\mathbf{n}}_{k,i},
\end{align}
where $\tilde{\mathbf{n}}_{k,i}$ is the $i$-th column of $\tilde{\mathbf{N}}_k$. Using the SNL-LS method \cite{peng2023two}, we can estimate $\left( \sigma _{k,i},\xi _{k,i} \right)$ via a simple two-dimension search as 
\begin{align}\label{snlls}
( \hat{\sigma}_{k,i},\hat{\xi}_{k,i} ) =\mathrm{arg}\max \,\varpi _{k,i}|\left< \mathbf{c}_{k,i},\mathbf{a}_N\left( \sigma _{k,i},\xi _{k,i} \right) \right> |^2,
\end{align}
where $\varpi _{k,i}=\left< \mathbf{a}_N\left( \sigma _{k,i},\xi _{k,i} \right) ,\mathbf{a}_N\left( \sigma _{k,i},\xi _{k,i} \right) \right> ^{-1}$ and $\sigma _{k,i},\xi _{k,i}\in [ \small{-\frac{d}{\lambda _c},\frac{d}{\lambda _c}} ]$. Once the $\sigma _{k,i}$ and $\xi _{k,i}$ are estimated, the corresponding gains $\zeta _{k,i}$ can be obtained by $\zeta _{k,i}=\mathbf{a}_N\left( \sigma _{k,i},\xi _{k,i} \right) ^{\dagger}\mathbf{c}_{k,i}$. Due to the relatively low pilot overhead required for estimating the direct channels, we repeat $K$ times to obtain the estimate of $\mathbf{H}_{\mathrm{d},k},\forall k\in \mathcal{K}$. The overall estimation for direct channels is summarized in \textbf{Algorithm} \ref{alg1}. 
\begin{algorithm}[t]
	\caption{ Estimation of $\mathbf{H}_{\mathrm{d},k},\forall k\in \mathcal{K}$.}
	\label{alg1}
	\renewcommand{\algorithmicrequire}{\textbf{Input:}}
	\renewcommand{\algorithmicensure}{\textbf{Output:}}
	\begin{algorithmic}[1]
		\REQUIRE $\mathbf{C}_k$, $\mathbf{S}_k$ for $\forall k\in \mathcal{K}$, and $\mathbf{D}_1$, $\mathbf{D}_2$.  
		\FOR{$1\leqslant k\leqslant K$}
			\STATE  \textbf{Phase 1:} Estimate $\mathbf{D}_{Q,k}$ via the SOMP algorithm.
			\STATE Calculate the dictionary $\mathbf{F}=\mathbf{S}_{k}^{\mathrm{H}}\left( \mathbf{D}_1\otimes \mathbf{D}_2 \right)$.
			\STATE Initialize $\mathbf{R}_0=\mathbf{C}_{k}^{\mathrm{H}}$, $\Xi _0=\emptyset$ and $i=1$.
			\REPEAT
			\STATE Calculate $\mathbf{\Psi }=\mathbf{F}^{\mathrm{H}}\mathbf{R}_{i-1}$.
			\STATE $b_i=\,\mathop {\mathrm{arg}\max} \limits_{b_i=1,...,B_1B_2}\mathrm{diag}\{\mathbf{\Psi \Psi }^{\mathrm{H}}\}$.
			\STATE $\Xi _i=\Xi _{i-1}\cup b_i$.
			\STATE Calculate $\mathbf{w}_i=( \mathbf{F}_{( :,\Xi _i )}^{\mathrm{H}}\mathbf{F}_{( :,\Xi _i )} ) ^{-1}\mathbf{F}_{( :,\Xi _i )}^{\mathrm{H}}\mathbf{C}_{k}^{\mathrm{H}}$.
			\STATE $\mathbf{R}_i=\mathbf{C}_{k}^{\mathrm{H}}-\mathbf{F}_{\left( :,\Xi _i \right)}\mathbf{w}_i$.
			\STATE $i=i+1$.
			\UNTIL{$|| \mathbf{R}_{i} || _F\leqslant $ threshold.}
			\STATE Obtain the estimate of $\rho _{k,i}$ and $\kappa _{k,i}$ for $\forall i$ via \eqref{b1b2}.
			\STATE  \textbf{Phase 2:} Estimate the remaining parameters by using the two-dimensional SNL-LS algorithm.
			\STATE Calculate $\tilde{\mathbf{C}}_k=\mathbf{C}_k( \hat{\mathbf{D}}_{Q,k}^{\mathrm{H}}\mathbf{S}_k ) ^{\dagger}$.
			\FOR{$1\leqslant i\leqslant \hat{I}_k$}
			\STATE Extract the $i$-th column of $\tilde{\mathbf{C}}_k$: $\mathbf{c}_{k,i}=[ \tilde{\mathbf{C}}_k ] _{( :,i )}$.
			\STATE Find $( \hat{\sigma}_{k,i},\hat{\xi}_{k,i})$ via \eqref{snlls}.
			\STATE Obtain the channel gains $\hat{\zeta}_{k,i}=\mathbf{a}_N( \hat{\sigma}_{k,i},\hat{\xi}_{k,i}) ^{\dagger}\mathbf{c}_{k,i}$.
			\ENDFOR
			\STATE Obtain the estimates:\\
			$\qquad\hat{\mathbf{D}}_{N,k}=[\mathbf{a}_N(\hat{\sigma}_{k,1},\hat{\xi}_{k,1}),...,\mathbf{a}_N(\hat{\sigma}_{k,\hat{I}_k},\hat{\xi}_{k,\hat{I}_k})]$,\\
			$\qquad\hat{\mathbf{Z}}_k=\mathrm{Diag}\{ \hat{\zeta}_{k,1},\ldots,\hat{\zeta}_{k,\hat{I}_k} \}$.
		\ENDFOR
		\ENSURE $\hat{\mathbf{H}}_{\mathrm{d},k}=\hat{\mathbf{D}}_{N,k}\hat{\mathbf{Z}}_k\hat{\mathbf{D}}_{Q,k}^{\mathrm{H}},\forall k\in \mathcal{K}$.   
	\end{algorithmic}
\end{algorithm}

\subsection{Stage ${\uppercase\expandafter{\romannumeral2}}$: Estimation of the AoDs of the user-RIS channel}\label{juti}
 \textit{1)} The protocol design in Segment $\mathrm{\uppercase\expandafter{\romannumeral1}}$ simultaneously supports estimation in Stage $\mathrm{\uppercase\expandafter{\romannumeral2}}$ for the AoDs of the user-RIS channel, i.e. $\left\{ \left( \eta _{k,j},\chi _{k,j} \right) \right\} _{j=1}^{J_k}$ for $\forall k$.
We denote the pilot overhead required for estimation in Stage $\mathrm{\uppercase\expandafter{\romannumeral2}}$ by $T_{\mathrm{AoD},k}$, $\left( T_{\mathrm{AoD},k} \leqslant T_{\mathrm{f},k} \right)$. Naturally, during the pilot transmission, $T_{\mathrm{f},k}$ is set to $T_{\mathrm{f},k}=\max \left\{ T_{\mathrm{d},k},T_{\mathrm{AoD},k} \right\}$. The received signal processing for Stage $\mathrm{\uppercase\expandafter{\romannumeral2}}$ is shown in the lower half of Fig. \ref{figfirst}.

The designed pilot sequences and RIS phase shifts enable the signal recombination and projection processing.
Specifically, stacking the received signals in time slot $t\left( 1\leqslant t\leqslant T_{\mathrm{AoD},k} \right)$ within time block $1$ and time block $2$ as $\mathbf{Y}_{\mathrm{f},k,t}=[ \mathbf{y}_{k,t}^{1},\mathbf{y}_{k,t}^{2} ]$, we have
\begin{align}
\mathbf{Y}_{\mathrm{f},k,t}=\mathbf{H}_{\mathrm{d},k}\mathbf{s}_{k,t}\mathbf{1}_{2}^{\mathrm{T}}+\left( \mathbf{s}_{k,t}^{\mathrm{T}}\otimes \mathbf{I}_N \right) \mathbf{G}_{k}\mathbf{E}_{\mathrm{f}}+\mathbf{N}_{\mathrm{f},k,t},
\end{align}
where $\mathbf{E}_{\mathrm{f}}=[ -\mathbf{e}_1,\mathbf{e}_1 ]$ and $\mathbf{N}_{\mathrm{f},k,t}=[ \mathbf{n}_{k,t}^{1},\mathbf{n}_{k,t}^{2} ]$. The orthogonal complement space of a two-dimensional vector is the span of its orthogonal vector with a unit length. Thus, we utilize the orthogonal vector of $\mathbf{1}_{2}$, i.e., $\mathbf{v}=[ -\sqrt{2}/2,\sqrt{2}/2 ] ^{\mathrm{T}}$ to obtain $\tilde{\mathbf{y}}_{\mathrm{f},k,t}=\mathbf{Y}_{\mathrm{f},k,t}\mathbf{v}/\sqrt{2}$ given by
\begin{align}\label{yv}
\tilde{\mathbf{y}}_{\mathrm{f},k,t}&=\left( \mathbf{H}_{\mathrm{d},k}\mathbf{s}_{k,t}\mathbf{1}_{2}^{\mathrm{T}}+\left( \mathbf{s}_{k,t}^{\mathrm{T}}\otimes \mathbf{I}_N \right) \mathbf{G}_{k}\mathbf{E}_{\mathrm{f}}+\mathbf{N}_{\mathrm{f},k,t} \right) \mathbf{v}/\sqrt{2}\nonumber
\\
&=\left( \mathbf{s}_{k,t}^{\mathrm{T}}\otimes \mathbf{I}_N \right) \mathbf{G}_{k}\mathbf{e}_1+\tilde{\mathbf{n}}_{\mathrm{f},k,t}\nonumber
\\
&=\mathbf{H}\mathrm{Diag}\left\{ \mathbf{e}_1 \right\} \mathbf{H}_{\mathrm{r},k}\mathbf{s}_{k,t}+\tilde{\mathbf{n}}_{\mathrm{f},k,t},
\end{align}
where we use $\left( \mathbf{C}^{\mathrm{T}}\otimes \mathbf{A} \right) \mathrm{vec}\left( \mathbf{B} \right) =\mathrm{vec}\left( \mathbf{ABC} \right)$ and $\left( \mathbf{C}^{\mathrm{T}}\diamond \mathbf{A} \right) \mathbf{b}=\mathrm{vec}\left( \mathbf{A}\mathrm{Diag}\left\{ \mathbf{b} \right\} \mathbf{C} \right)$ to obtain the third equation, and $\tilde{\mathbf{n}}_{\mathrm{f},k,t}=\small{-\frac{1}{2}}\mathbf{n}_{k,t}^{1}+\small{\frac{1}{2}}\mathbf{n}_{k,t}^{2}$. As observed in \eqref{yv}, the components of the direct channels are completely eliminated in the obtained equivalent signal $\tilde{\mathbf{y}}_{\mathrm{f},k,t}$ for estimating the AoDs of the user-RIS channels, which avoids the error propagation from the direct channel estimation to the estimation of cascaded channels. Furthermore, we note that the equivalent noise power of cascaded channel is reduced to half of its original value, as evident from the third equation in \eqref{yv}.

Applying the same signal pre-processing as \eqref{yv} to all $T_{\mathrm{AoD},k}$ pairs of time slots in the first two blocks and stacking the overall $T_{\mathrm{AoD},k}$ received signals, we obtain $\tilde{\mathbf{Y}}_{\mathrm{f},k}$ as 
\begin{align}\label{YFK}
\tilde{\mathbf{Y}}_{\mathrm{f},k}=\mathbf{H}\mathrm{Diag}\left\{ \mathbf{e}_1 \right\} \mathbf{H}_{\mathrm{r},k}\mathbf{S}_{\mathrm{AoD},k}+\tilde{\mathbf{N}}_{\mathrm{f},k},
\end{align}
where $\mathbf{S}_{\mathrm{AoD},k}=[ \mathbf{s}_{k,1},\ldots,\mathbf{s}_{k,T_{\mathrm{AoD},k}} ]\in \mathbb{C} ^{Q_k\times T_{\mathrm{AoD},k}}$ is regarded as the equivalent pilot matrix transmitted by user $k$ in all $T_{\mathrm{AoD},k}$ time slots illustrated in Fig. \ref{fig1}.
$\tilde{\mathbf{N}}_{\mathrm{f},k}=[ \tilde{\mathbf{n}}_{\mathrm{f},k,1},\ldots,\tilde{\mathbf{n}}_{\mathrm{f},k,T_{\mathrm{AoD},k}} ]\in \mathbb{C} ^{N\times T_{\mathrm{AoD},k}}$ is the corresponding noise matrix. Combining with \eqref{h_r}, we take the conjugate transpose of \eqref{YFK} as
\begin{align}\label{YfkH}
\tilde{\mathbf{Y}}_{\mathrm{f},k}^{\mathrm{H}}=\mathbf{S}_{\mathrm{AoD},k}^{\mathrm{H}}\mathbf{A}_{Q,k}\mathbf{\Phi }_k+\tilde{\mathbf{N}}_{\mathrm{f},k}^{\mathrm{H}},
\end{align}
where $\mathbf{\Phi }_k\triangleq ( \mathbf{A}_N\mathbf{\Lambda A}_{M}^{\mathrm{H}}\mathrm{Diag}\left\{ \mathbf{e}_1 \right\} \mathbf{A}_{M,k}\mathbf{B}_k ) ^{\mathrm{H}}$. Based on the obtained \eqref{YfkH}, we can estimate the AoDs of user-RIS channels, i.e. $\mathbf{A}_{Q,k}$. Due to its resemblance to the problem associated with \eqref{CkH}, the problem associated with \eqref{YfkH} can be transformed into a sparse recovery problem for MMV case as, 
\begin{align}\label{YKHVAD}
\tilde{\mathbf{Y}}_{\mathrm{f},k}^{\mathrm{H}}=\mathbf{S}_{\mathrm{AoD},k}^{\mathrm{H}}\left( \mathbf{D}_1\otimes \mathbf{D}_2 \right) \tilde{\mathbf{\Phi}}_k+\tilde{\mathbf{N}}_{\mathrm{f},k}^{\mathrm{H}},
\end{align}
where $\mathbf{D}_1$ and $\mathbf{D}_2$ are overcomplete dictionary matrices defined in \eqref{CkHVAD} and $\tilde{\mathbf{\Phi}}_k\in \mathbb{C} ^{B_1B_2\times N}$ is a row-sparse matrix with $J_k$ non-zero rows. Similarly, utilizing the SOMP procedure as shown in the Algorithm \ref{alg1}, we can obtain the estimates of $\mathbf{A}_{Q,k}$, denoted by $\hat{\mathbf{A}}_{Q,k}$, and obtain the corresponding estimated AoDs $\left\{ \left( \hat{\eta}_{k,j},\hat{\chi}_{k,j} \right) \right\} _{j=1}^{J_k}$.

In Stage $\mathrm{\uppercase\expandafter{\romannumeral1}}$ and Stage $\mathrm{\uppercase\expandafter{\romannumeral2}}$, the two time blocks of Segment $\mathrm{\uppercase\expandafter{\romannumeral1}}$ are used to estimate the parameters of the user-BS channel and the AoDs of the user-RIS channel with low pilot overhead. Users transmit pilot signals sequentially, enabling the acquisition of all users' parameters by repeating the process $K$ times.

Moreover, in Stage $\mathrm{\uppercase\expandafter{\romannumeral1}}$ and $\mathrm{\uppercase\expandafter{\romannumeral2}}$, we completely separate the direct channel estimation and the cascaded channel estimation in the received signals from the time blocks of Segment $\mathrm{\uppercase\expandafter{\romannumeral1}}$. In Stage $\mathrm{\uppercase\expandafter{\romannumeral3}}$ of the following Section \ref{stage2}, we separate the direct channel estimation and the cascaded channel estimation in the received signals from the time blocks of Segment $\mathrm{\uppercase\expandafter{\romannumeral2}}$ through more flexible signal processing in order to estimate the remaining parameters of the cascaded channel.

\textit{2)} Before illustrating the details of Stage $\mathrm{\uppercase\expandafter{\romannumeral3}}$ of the proposed method, a brief introduction to the existing ON-OFF method is presented, which is known as a classical method for separating the direct and cascaded channel estimation. The mechanism of the ON-OFF method is summarized as follows.

First, switch off the RIS and the equivalent received signal for estimating the direct channels $\mathbf{H}_{\mathrm{d},k}$ is given by
\begin{align}\label{ydk}
	\mathbf{y}_{\mathrm{d},k,t} =\mathbf{H}_{\mathrm{d},k}\mathbf{s}_{k,t} +\mathbf{n}_{\mathrm{d},k},
\end{align}
where $\mathbf{n}_{\mathrm{d},k}$ is the noise vector. Then, switching on the RIS, the equivalent signal $\breve{\mathbf{y}}_k$ for cascaded channel estimation is obtained by subtracting the estimate of $\mathbf{H}_{\mathrm{d},k}$, denoted by $\hat{\mathbf{H}}_{\mathrm{d},k}$, from the received signal. Specifically, $\breve{\mathbf{y}}_k$ is given by 
\begin{align}\label{onoffy}
\breve{\mathbf{y}}_{k,t} &=\left( \Delta \mathbf{H}_{\mathrm{d},k}+\mathbf{H}\mathrm{Diag}\left\{ \mathbf{e}_t \right\} \mathbf{H}_{\mathrm{r},k} \right) \mathbf{s}_{k,t}+\mathbf{n}_{k,t}  \nonumber
\\
&=\mathbf{H}\mathrm{Diag}\left\{ \mathbf{e}_t \right\} \mathbf{H}_{\mathrm{r},k}\mathbf{s}_{k,t} +\breve{\mathbf{n}}_{k,t},
\end{align}
where $\breve{\mathbf{n}}_{k,t} =\mathbf{n}_{k,t} +\Delta \mathbf{H}_{\mathrm{d},k}$, and $\Delta \mathbf{H}_{\mathrm{d},k}\triangleq (\mathbf{H}_{\mathrm{d},k}-\hat{\mathbf{H}}_{\mathrm{d},k})$ represents the inevitable direct-to-cascaded error propagation.
Since $\breve{\mathbf{n}}_k$ and $\Delta \mathbf{H}_{\mathrm{d},k}$ originate from different sources, they can be assumed uncorrelated. Therefore, $Cov\left( \breve{\mathbf{n}}_{k,t} \right) =Cov\left( \mathbf{n}_{k,t} \right) +Cov\left( \Delta \mathbf{H}_{\mathrm{d},k} \right) =\left( \delta ^2+\epsilon \right) \mathbf{I}$, $\epsilon >0$ 
It significantly deteriorates the estimation of the cascaded channel. 
Additionally, in the ON-OFF method, switching the RIS on and off frequently leads to power loss in practice \cite{Phys}.
By contrast, the proposed signal pre-processing method addresses the above-mentioned drawbacks of the ON-OFF method. The details for the proposed method are presented in Section \ref{stage2}.

\subsection{Stage ${\uppercase\expandafter{\romannumeral3}}$: Estimation of the AoAs of the RIS-BS channel and the remaining parameters of the cascaded channel}\label{stage2}


In Stage $\mathrm{\uppercase\expandafter{\romannumeral3}}$, as shown in Fig. \ref{fig1}, we utilize time blocks of both Segment $\mathrm{\uppercase\expandafter{\romannumeral1}}$ and $\mathrm{\uppercase\expandafter{\romannumeral2}}$, i.e., from time block $1$ to time block $\left( X_k+2 \right)$ to estimate the AoAs of the RIS-BS channel and the remaining parameters of the cascaded channel. We propose a novel signal pre-processing method that not only avoids switching on-off the RIS, which is hardware cost-effective, but also completely eradicates the direct-to-cascaded channel estimation error effect for estimating the AoAs of the RIS-BS channel and the remaining parameters of the cascaded channel.

\subsubsection{Signal pre-processing based on orthogonal complement space projection and received signal recombination}
\begin{figure}
	\centering
	\includegraphics[width=8.5cm]{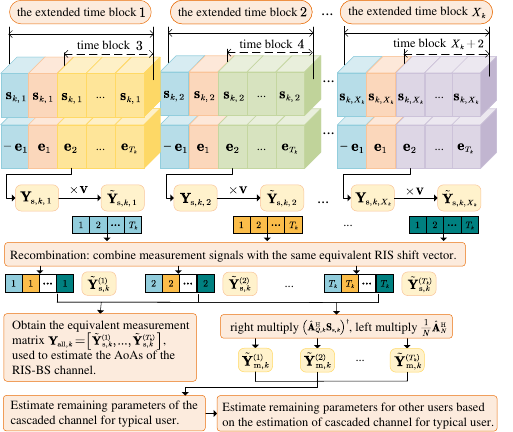}
	\caption{The schematic diagram of signal pre-processing in Stage $\mathrm{\uppercase\expandafter{\romannumeral3}}$.}
	\label{figchongzu} 
\end{figure}
We take the received signal at the BS in time slot $t$, denoted by $\mathbf{y}_{k,t}$, as an example to show the signal pre-processing technique so as to obtain the equivalent received signals for estimating the AoAs of the RIS-BS channel and the remaining parameters of the cascaded channels. Recalling \eqref{yk}, $\mathbf{y}_{k,t}$ is expressed as
\begin{align}\label{ykt}
	\mathbf{y}_{k,t} &=\mathbf{H}_{\mathrm{d},k}\mathbf{s}_{k,t} +\mathbf{H}\mathrm{Diag}\left\{ \mathbf{e}_t \right\} \mathbf{H}_{\mathrm{r},k}\mathbf{s}_{k,t} +\mathbf{n}_{k,t} \nonumber
	\\
	&=\mathbf{H}_{\mathrm{d},k}\mathbf{s}_{k,t} +\left( \mathbf{s}_{k,t}^{\mathrm{T}} \otimes \mathbf{I}_N \right) \mathbf{G}_{k}\mathbf{e}_t+\mathbf{n}_{k,t}.
\end{align}
As depicted in Fig. \ref{figchongzu}, the Segment $\mathrm{\uppercase\expandafter{\romannumeral2}}$ deliberately exploits variations in the vectors of RIS, while maintaining identical pilots within each time block, to generate independent measurements. Specifically, first, we stack the received signals within an extended time block so as to obtain the signal matrices during the given extended time block. The extended time block consists of two time slots from the time blocks of Segment $\mathrm{\uppercase\expandafter{\romannumeral1}}$ and overall $(T_k-1)$ time slots from the subsequent time blocks of Segment $\mathrm{\uppercase\expandafter{\romannumeral2}}$. The index of the selected time slot in the first two time blocks, as well as the index of the subsequent time blocks, both correspond to the index of the extended time block.
For instance, the received signal $\mathbf{y}_{k,t}$ in the time slots of the extended time block $(x_k+2)$ $(1\leqslant x_k\leqslant X_k)$ are stacked and thus the received matrix during the extended time block $(x_k+2)$ is obtained, denoted by
$\mathbf{Y}_{\mathrm{s},k,x_k}\in \mathbb{C} ^{N\times \left( T_k+1 \right)}$ as
\begin{align}\label{YSKX}
\mathbf{Y}_{\mathrm{s},k,x_k}&=\mathbf{H}_{\mathrm{d},k}\mathbf{s}_{k,x_k}\mathbf{1}_{T_k+1}^{\mathrm{T}}+\left( \mathbf{s}_{k,x_k}^{\mathrm{T}}\otimes \mathbf{I}_N \right) \mathbf{G}_{k}\mathbf{E}_{\mathrm{s},k}+\mathbf{N}_{\mathrm{s},k,x_k},
\end{align}
where $\mathbf{E}_{\mathrm{s},k}\triangleq \left[ -\mathbf{e}_1,\mathbf{e}_1,\mathbf{e}_2,\ldots,\mathbf{e}_{T_k} \right]$ is the matrix of RIS and $\mathbf{N}_{\mathrm{s},k,x_k}=[ \mathbf{n}_{k,x_k}^{1},\mathbf{n}_{k,x_k}^{2}, \mathbf{n}_{k,1}^{x_k+2},\ldots,\mathbf{n}_{k,T_k-1}^{x_k+2} ]$ is the noise matrix during the extended time block $(x_k+2)$.

To remove the component of the direct channel in $\mathbf{Y}_{\mathrm{s},k,x_k}$, the following signal pre-processing is performed based on the orthogonal complement space projection.
	Let $\boldsymbol{U}$ denote the space spanned by the vector $\mathbf{1}_{T_k+1}$ in \eqref{YSKX}, which is a subspace of the $(T_k+1)$-dimensional linear space $\mathbb{C} ^{T_k+1}$. $\boldsymbol{U}^{\bot}$ is the orthogonal complement of $\boldsymbol{U}$ satisfying
\begin{align}
	\begin{cases}
		\,\,\boldsymbol{U}^{\bot}=\left\{ \left. \mathbf{f} \right|\mathbf{f}^{\mathrm{H}}\mathbf{g}=0,\forall \mathbf{g}\in \boldsymbol{U} \right\}\\
		\,\,\mathrm{dim}( \boldsymbol{U}) +\mathrm{dim}( \boldsymbol{U}^{\bot} ) =T_k+1.\\
	\end{cases}
\end{align}
Thus, a column-unitary matrix $\mathbf{V}\in \mathbb{C} ^{\left( T_k+1 \right) \times T_k}$ is constructed, whose column vectors are constructed from the subspace $\boldsymbol{U}^{\bot}$.
Now projecting $\mathbf{Y}_{\mathrm{s},k,x_k}$ onto the orthogonal complement space of $\mathbf{1}_{T_k+1}$, we have $\tilde{\mathbf{Y}}_{\mathrm{s},k,x_k}=\mathbf{Y}_{\mathrm{s},k,x_k}\mathbf{V}$ as
\begin{align}\label{YV}
\tilde{\mathbf{Y}}_{\mathrm{s},k,x_k}&=\left( \mathbf{s}_{k,x_k}^{\mathrm{T}}\otimes \mathbf{I}_N \right) \mathbf{G}_{k}\mathbf{E}_{\mathrm{s},k}\mathbf{V}+\mathbf{N}_{\mathrm{s},k,x_k}\mathbf{V}\nonumber
\\
&\triangleq \left( \mathbf{s}_{k,x_k}^{\mathrm{T}}\otimes \mathbf{I}_N \right) \mathbf{G}_{k}\tilde{\mathbf{E}}_{\mathrm{s},k}+\tilde{\mathbf{N}}_{\mathrm{s},k,x_k},
\end{align}
where 
$\tilde{\mathbf{E}}_{\mathrm{s},k}\triangleq \mathbf{E}_{\mathrm{s},k}\mathbf{V}=\left[ \boldsymbol{\phi }_1,\ldots,\boldsymbol{\phi }_{T_k} \right]\in \mathbb{C} ^{M\times T_k}$ is regarded as the equivalent matrix of the RIS and $\tilde{\mathbf{N}}_{\mathrm{s},k,x_k}$ is the corresponding projected noise matrix.
The statistical characteristics of noise remains unchanged due to Lemma \ref{lemma}. 
\begin{lemma}\label{lemma}
	\textit{After the subspace projection, the noise power remains unchanged. Specifically, the elements in $\tilde{\mathbf{N}}_{\mathrm{s},k,x_k}$ are independent and identically distributed with $\mathcal{C} \mathcal{N} \left( 0,\delta ^2 \right)$. }
\end{lemma}

\begin{proof}
We have $\mathrm{vec}( \tilde{\mathbf{N}}_{\mathrm{s},k,x_k}) =( \mathbf{V}^{\mathrm{T}}\otimes \mathbf{I}_N ) \mathrm{vec}\left( \mathbf{N}_{\mathrm{s},k,x_k} \right)$, 
	where $\mathrm{vec}( \tilde{\mathbf{N}}_{\mathrm{s},k,x_k})$ still follows the complex Gaussian distribution, since it can be expressed as a linear transformation of the complex Gaussian random vector $\mathrm{vec}(\mathbf{N}_{\mathrm{s},k,x_k}) $.
	Specifically, using the property of circularly-symmetric Gaussian vectors, the covariance matrix of $\mathrm{vec}( \tilde{\mathbf{N}}_{\mathrm{s},k,x_k} )$ is given by
	\begin{align}\label{bu}
		Cov( \mathrm{vec}(\tilde{\mathbf{N}}_{\mathrm{s},k,x_k} ))&=\delta ^2( \mathbf{V}^{\mathrm{T}}\otimes \mathbf{I}_N ) ( \mathbf{V}^*\otimes \mathbf{I}_N )\nonumber
		\\
		&=\delta ^2( \mathbf{V}^{\mathrm{T}}\mathbf{V}^*) \otimes \mathbf{I}_N=\delta ^2\mathbf{I}_{NT_k}.
	\end{align}
	From \eqref{bu}, the variance of each element in $\tilde{\mathbf{N}}_{\mathrm{s},k,x_k}$ is still $\delta ^2$, hence $\mathrm{vec}( \tilde{\mathbf{N}}_{\mathrm{s},k,x_k}) \sim \mathcal{C} \mathcal{N} \left( 0,\delta ^2\mathbf{I}_{NT_k} \right)$. Thus, the noise power remains unchanged after projection.
\end{proof}

Following the derivation similar to \eqref{yk} and \eqref{yv}, we extract the $w$-th $(1\leqslant w\leqslant T_k)$ column of the matrix $\tilde{\mathbf{Y}}_{\mathrm{s},k,x_k}$ as
\begin{align}\label{Yw}
&[\tilde{\mathbf{Y}}_{\mathrm{s},k,x_k}]_{( :,w )}=( \mathbf{s}_{k,x_k}^{\mathrm{T}}\otimes \mathbf{I}_N ) \mathbf{G}_{k}\boldsymbol{\phi }_w+[\tilde{\mathbf{N}}_{\mathrm{s},k,x_k}]_{\left( :,w \right)}\nonumber
\\
&=\mathbf{H}\mathrm{Diag}\left\{ \boldsymbol{\phi }_w \right\} \mathbf{H}_{\mathrm{r},k}\mathbf{s}_{k,x_k}+[\tilde{\mathbf{N}}_{\mathrm{s},k,x_k}]_{\left( :,w \right)}.
\end{align}
Repeating the same signal pre-processing steps as shown from \eqref{YSKX} to \eqref{Yw} for the overall $X_k$ extended time blocks, i.e., the extended time block 3 to the extended time block $(X_k+2)$, $\{ [ \tilde{\mathbf{Y}}_{\mathrm{s},k,x_k}] _{\left( :,w \right)} \}$ $(1\leqslant x_k\leqslant X_k)$ can be obtained. Stacking them, $\tilde{\mathbf{Y}}_{\mathrm{s},k}^{\left( w \right)}\in \mathbb{C} ^{N\times X_k}$ is given by,
\begin{align}\label{YWSTACK}
\tilde{\mathbf{Y}}_{\mathrm{s},k}^{\left( w \right)}&=[ [ \tilde{\mathbf{Y}}_{\mathrm{s},k,1} ] _{\left( :,w \right)},[ \tilde{\mathbf{Y}}_{\mathrm{s},k,2} ] _{\left( :,w \right)},\ldots,[ \tilde{\mathbf{Y}}_{\mathrm{s},k,X_k} ] _{\left( :,w \right)}] \nonumber
\\
&=\mathbf{H}\mathrm{Diag}\left\{ \boldsymbol{\phi }_w \right\} \mathbf{H}_{\mathrm{r},k}\mathbf{S}_{\mathrm{s},k}+\tilde{\mathbf{N}}_{\mathrm{s},k}^{\left( w \right)},
\end{align}
where $\mathbf{S}_{\mathrm{s},k}=[ \mathbf{s}_{k,1},\ldots,\mathbf{s}_{k,X_k} ]$ and $\mathbf{N}_{\mathrm{s},k}^{\left( w \right)}=[ [ \tilde{\mathbf{N}}_{\mathrm{s},k,1} ] _{\left( :,w \right)},\ldots,[ \tilde{\mathbf{N}}_{\mathrm{s},k,X_k} ] _{\left( :,w \right)} ]$ are the corresponding pilot matrix and noise matrix of the received signal $\tilde{\mathbf{Y}}_{\mathrm{s},k}^{\left( w \right)}$.

Finally, we have completed the signal pre-processing of the received signal $\mathbf{y}_{k,t}$ in \eqref{ykt} and obtain the equivalent signal matrices $\tilde{\mathbf{Y}}_{\mathrm{s},k}^{\left( w \right)}$ $(1\leqslant w_k\leqslant T_k)$. In the following, we estimate the AoAs of the RIS-BS channel i.e., the AoAs of the RIS-BS channel $\left\{ \left( \psi _l,\nu _l \right) \right\} _{l=1}^{L}$, and the remaining parameters of the cascaded channel, i.e., $\left\{ \left( \omega _l-\varphi _{k,j} \right) \right\} _{j=1}^{J_k},\left\{ \left( \mu _l-\theta _{k,j} \right) \right\} _{j=1}^{J_k},\{ \alpha _{l}^{*}\beta _{k,j}^{*} \} _{j=1}^{J_k}$ for $\forall l\in \left\{ 1,\ldots,L \right\}$, based on the $\tilde{\mathbf{Y}}_{\mathrm{s},k}^{\left( w \right)}$.
\subsubsection{Estimation of the common AoAs of the RIS-BS channel}\label{UPA-DFT}
Combining \eqref{H} with \eqref{YWSTACK}, we have
\begin{align}\label{Ywh}
\tilde{\mathbf{Y}}_{\mathrm{s},k}^{\left( w \right)}=\mathbf{A}_N\mathbf{\Lambda A}_{M}^{\mathrm{H}}\mathrm{Diag}\left\{ \boldsymbol{\phi }_w \right\} \mathbf{H}_{\mathrm{r},k}\mathbf{S}_{\mathrm{s},k}+\tilde{\mathbf{N}}_{\mathrm{s},k}^{\left( w \right)}.
\end{align}
It is noted that the equivalent signal matrix in \eqref{Ywh} bears resemblance to (9) in \cite{peng2022channel}. The AoAs of the RIS-BS channel $\mathbf{A}_N$ can be estimated by using the UPA-DFT method.

The accuracy of estimating $\mathbf{A}_N$ is crucial for the remaining parameters estimation including $\{ ( \omega _l-\varphi _{k,j} ) \} _{j=1}^{J_k}$, $\{ \left( \mu _l-\theta _{k,j} \right) \} _{j=1}^{J_k}$, $\{ \alpha _{l}^{*}\beta _{k,j}^{*} \} _{j=1}^{J_k}$ for $\forall l\in \left\{ 1,\ldots,L \right\}$. The $\mathbf{A}_N$ exists in $\tilde{\mathbf{Y}}_{\mathrm{s},k}^{\left( w \right)}\left( 1\leqslant w\leqslant T_k \right) $ and is shared by all users' cascaded channels. Therefore, to enhance the estimation accuracy, we stack all $\tilde{\mathbf{Y}}_{\mathrm{s},k}^{\left( w \right)}\left( 1\leqslant w\leqslant T_k \right)$ to obtain 
$\mathbf{Y}_{\mathrm{all},k}=[ \tilde{\mathbf{Y}}_{\mathrm{s},k}^{\left( 1 \right)},\ldots,\tilde{\mathbf{Y}}_{\mathrm{s},k}^{\left( T_k \right)}] \in \mathbb{C} ^{N\times X_kT_k}$, and further, stack $\mathbf{Y}_{\mathrm{all},k}$ of all users to get
$\mathbf{Y}_{\mathrm{all}}=\left[ \mathbf{Y}_{\mathrm{all},1},\ldots,\mathbf{Y}_{\mathrm{all},K} \right] \in \mathbb{C} ^{N\times \sum_{k=1}^K{X_kT_k}}$.
Then, the UPA-DFT method is used to obtain the estimates of $\mathbf{A}_N$, denoted by $\hat{\mathbf{A}}_N$ based on $\mathbf{Y}_{all}$. This operation makes full use of a larger amount of information, which reduces the error propagation from $\mathbf{A}_N$ to the estimation of remaining parameters.

Then, based on the estimated $\hat{\mathbf{A}}_{Q,k}$ and $\hat{\mathbf{A}}_N$, we obtain the measurement matrix $\tilde{\mathbf{Y}}_{\mathrm{m},k}^{\left( w \right)}\in \mathbb{C} ^{L\times J_k}$ by processing $\tilde{\mathbf{Y}}_{\mathrm{s},k}^{\left( w \right)}$ as
\begin{align}\label{AYAS}
	\tilde{\mathbf{Y}}_{\mathrm{m},k}^{\left( w \right)}
	&=\frac{1}{N}\hat{\mathbf{A}}_{N}^{\mathrm{H}}\tilde{\mathbf{Y}}_{\mathrm{s},k}^{\left( w \right)}( \hat{\mathbf{A}}_{Q,k}^{\mathrm{H}}\mathbf{S}_{\mathrm{s},k}) ^{\dagger}\nonumber
	\\
	&=\frac{1}{N}\hat{\mathbf{A}}_{N}^{\mathrm{H}}\left[  \right. ( \hat{\mathbf{A}}_N+\Delta \mathbf{A}_N ) \mathbf{\Lambda A}_{M}^{\mathrm{H}}\mathrm{Diag}\left\{ \boldsymbol{\phi }_w \right\} \mathbf{A}_{M,k}\mathbf{B}_k\nonumber
	\\
	&\,\,\,\,\,( \hat{\mathbf{A}}_{Q,k}+\Delta \mathbf{A}_{Q,k}) \mathbf{S}_{\mathrm{s},k}+\tilde{\mathbf{N}}_{\mathrm{s},k}^{\left( w \right)} \left.  \right] ( \hat{\mathbf{A}}_{Q,k}^{\mathrm{H}}\mathbf{S}_{\mathrm{s},k}) ^{\dagger}\nonumber
	\\
	&=\mathbf{\Lambda A}_{M}^{\mathrm{H}}\mathrm{Diag}\left\{ \boldsymbol{\phi }_w \right\} \mathbf{A}_{M,k}\mathbf{B}_k+\tilde{\mathbf{N}}_{\mathrm{m},k}^{\left( w \right)},
\end{align}
where $\Delta \mathbf{A}_N$ and $\Delta \mathbf{A}_{Q,k}$ are the estimation errors of $\mathbf{A}_N$ and $ \mathbf{A}_{Q,k}$, respectively, and $\tilde{\mathbf{N}}_{\mathrm{m},k}^{\left( w \right)}$ is the equivalent noise matrix including noise and the errors from estimating $\mathbf{A}_N$ and $ \mathbf{A}_{Q,k}$.
In the following subsection, with the pre-processed signal matrix $\tilde{\mathbf{Y}}_{\mathrm{m},k}^{\left( w \right)}$ for $\forall k$, we first present the signal pre-processing for the estimation of remaining parameters for the cascaded channels. Then, we provide the estimation of remaining parameters for a typical user, denoted as user 1 for convenience. Finally, we introduce the estimation of remaining parameters for other users, denote as user $k$ for $k \ne 1$, by utilizing the scaling property.
\begin{algorithm}
	\caption{Estimation of the Remaining Parameters of the Cascaded Channel for a Typical User.}
	\label{alg2}
	\renewcommand{\algorithmicrequire}{\textbf{Input:}}
	\renewcommand{\algorithmicensure}{\textbf{Output:}}
	\begin{algorithmic}[1]
		\REQUIRE $\{ \tilde{\mathbf{Y}}_{\mathrm{vir},1,j} \} _{j=1}^{J_1}$ and $\mathbf{A}_1$, $\mathbf{A}_2$. 
		\STATE  \textbf{Phase 1:} Estimation for the virtual user $\left( 1,1 \right)$.
		\STATE $\tilde{\mathbf{y}}_{u}^{\left( 1,1 \right)}=\left[ \tilde{\mathbf{Y}}_{\mathrm{vir},1,j}^{\mathrm{H}} \right] _{\left( :,u \right)}$ according to \eqref{yu}.
		\STATE Calculate the equivalent dictionary $\mathbf{\Upsilon }=\tilde{\mathbf{E}}_{\mathrm{s},1}^{\mathrm{H}}\left( \mathbf{A}_1\otimes \mathbf{A}_2 \right)$.
		\STATE Find $a=\mathop {\mathrm{arg}\max} \limits_{i=1,...,F_1F_2}\| \mathbf{\Upsilon }_{\left( :,i \right)}^{\mathrm{H}}\tilde{\mathbf{y}}_{u}^{\left( 1,1 \right)} \|$.
		\STATE Obtain the estimate $\widehat{\alpha _{u}^{*}\beta _{1,1}^{*}}$.
		\STATE Calculate $a_{F_1}=\lceil \frac{a}{F_2} \rceil ,\,\,a_{F_2}=a-F_2\left( a_{F_1}-1 \right)$.
		\STATE Find the estimates  $\widehat{\omega _u-\varphi _{1,1}}$ and $\widehat{\mu _u-\theta _{1,1}}$ by referring to the indices from dictionary matrices $\mathbf{A}_1$ and $\mathbf{A}_2$.
		\FOR{$1\leqslant l\leqslant L\,\,\left( l\ne u \right)$}
		\STATE Search $\Delta \hat{\omega}_l$ and $\Delta \hat{\mu}_l$ within $[ -2\small{\frac{d}{\lambda _c}},2\small{\frac{d}{\lambda _c}} ] $ as
		\begin{small}
			\begin{align}\label{deltwv}
				\left( \Delta \hat{\omega}_l,\Delta \hat{\mu}_l \right) =\mathrm{arg}\max \small{\frac{| \left< \tilde{\mathbf{y}}_{l}^{( 1,j )},\mathbf{v}( \Delta \omega _l,\Delta \mu _l ) \right> |^2}{| \mathbf{v }( \Delta \omega _l,\Delta \mu _l ) |^2}},
			\end{align}
		\end{small}
		\STATE Calculate $\hat{\gamma}_l=\mathbf{v }\left( \Delta \hat{\omega}_l,\Delta \hat{\mu}_l \right) ^{\dagger}\tilde{\mathbf{y}}_{l}^{\left( 1,1 \right)}$.
		\STATE Calculate $\widehat{\omega _l-\varphi _{1,1}}=\widehat{\omega _u-\varphi _{1,1}}+\Delta \hat{\omega}_l$, $\widehat{\mu _l-\theta _{1,1}}=\widehat{\mu _u-\theta _{1,1}}+\Delta \hat{\mu}_l$ and $\widehat{\alpha _{l}^{*}\beta _{1,1}^{*}}=\hat{\gamma}_l\widehat{\alpha _{u}^{*}\beta _{1,1}^{*}}$.
		\ENDFOR
		\STATE Obtain  $\{ ( \widehat{\omega _l-\varphi _{1,1}} ) \} _{l=1}^{L}$, $\{ ( \widehat{\mu _l-\theta _{1,1} }) \} _{l=1}^{L}$, $\{ \widehat{\alpha _{l}^{*}\beta _{1,1}^{*} }\} _{l=1}^{L}$,\\
		$\{ \widehat{\Delta \omega _l} \} _{l=1}^{L}$, $\{ \widehat{\Delta \mu _l }\} _{l=1}^{L}$ and $\left\{ \widehat{\gamma _l} \right\} _{l=1}^{L}$.
		\STATE  \textbf{Phase 2:} Estimation for the virtual user $\left( 1,j \right), \left( j\ne 1 \right)$.
		\FOR{$2\leqslant j\leqslant J_1$}	
		\STATE Find $a_j=\mathop {\mathrm{arg}\max} \limits_{i=1,...,F_1F_2}\| \mathbf{\Upsilon }_{\left( :,i \right)}^{\mathrm{H}}\tilde{\mathbf{y}}_{u}^{\left( 1,j \right)} \|$.
		\STATE  Obtain the estimates $\widehat{\omega _u-\varphi _{1,j}}$, $\widehat{\mu _u-\theta _{1,j}}$ and $\widehat{\alpha _{u}^{*}\beta _{1,j}^{*}}$.
		\STATE Calculate $\widehat{\omega _l-\varphi _{1,j}}=\widehat{\omega _u-\varphi _{1,j}}+\Delta \hat{\omega}_l$, $\widehat{\mu _l-\theta _{1,j}}=\widehat{\mu _u-\theta _{1,j}}+\Delta \hat{\mu}_l$, $\widehat{\alpha _{l}^{*}\beta _{1,j}^{*}}=\hat{\gamma}_l\widehat{\alpha _{u}^{*}\beta _{1,j}^{*}}$ for $1\leqslant l\leqslant L$.
		\ENDFOR
		\STATE Obtain  $\{ ( \widehat{\omega _l-\varphi _{1,j}} ) \} _{l=1}^{L}$, $\{ ( \widehat{\mu _l-\theta _{1,j} }) \} _{l=1}^{L}$, $\{ \widehat{\alpha _{l}^{*}\beta _{1,j}^{*} }\} _{l=1}^{L}$ for $j\in \left\{ 1,\ldots,J_1 \right\}$.
		\ENSURE $\hat{\mathbf{G}}_1$.   
	\end{algorithmic}
\end{algorithm}

\subsubsection{Signal processing based on EVSA for the remaining parameters estimation for cascaded channels}\label{typical}
According to the signal model in \eqref{AYAS}, signal processing is conducted based on the idea of EVSA. Specifically, we convert the channel parameter estimation of a multi-antenna user $k$ with $J_k$ paths to the equivalent estimation of $J_k$ single-antenna virtual users with a single path. This method allows us to consider each virtual user $\left( k,j \right) \,\,\left( k\in \left\{ 1,\ldots,K \right\} ,j\in \left\{ 1,\ldots,J_k \right\} \right)$ individually. Based on this equivalent conversion, we extract the $j$-th $\left( 1\leqslant j\leqslant J_k \right)$ column of $\tilde{\mathbf{Y}}_{\mathrm{m},k}^{\left( w \right)}$ to obtain $\tilde{\mathbf{y}}_{\mathrm{m},k,j}^{\left( w \right)}$ as
\begin{align}
\tilde{\mathbf{y}}_{\mathrm{m},k,j}^{\left( w \right)}
&=\mathbf{\Lambda A}_{M}^{\mathrm{H}}\mathrm{Diag}\left\{ \boldsymbol{\phi }_w \right\} \left[ \mathbf{A}_{M,k} \right] _{\left( :,j \right)}\beta _{k,j}+\tilde{\mathbf{n}}_{\mathrm{m},k,j}^{\left( w \right)}\nonumber
\\
&=\mathbf{\Lambda A}_{M}^{\mathrm{H}}\mathrm{Diag}\{ \left[ \mathbf{A}_{M,k} \right] _{\left( :,j \right)}\beta _{k,j} \} \boldsymbol{\phi }_w+\tilde{\mathbf{n}}_{\mathrm{m},k,j}^{\left( w \right)},
\end{align}
where $\tilde{\mathbf{n}}_{\mathrm{m},k,j}^{\left( w \right)}=[ \tilde{\mathbf{N}}_{\mathrm{m},k}^{\left( w \right)}] _{\left( :,j \right)}$ is the corresponding noise vector.
Then, we stack all $\tilde{\mathbf{y}}_{\mathrm{m},k,j}^{\left( w \right)}\left( 1\leqslant w\leqslant T_k \right)$ to obtain the equivalent measurement matrix $\tilde{\mathbf{Y}}_{\mathrm{vir},k,j}=[\tilde{\mathbf{y}}_{\mathrm{m},k,j}^{\left( 1 \right)},\ldots,\tilde{\mathbf{y}}_{\mathrm{m},k,j}^{\left( T_k \right)} ]\in \mathbb{C} ^{L\times T_k}$ for the virtual user $\left( k,j \right)$ given by
\begin{align}\label{Yvirkj}
\tilde{\mathbf{Y}}_{\mathrm{vir},k,j}
&=\mathbf{\Lambda A}_{M}^{\mathrm{H}}\mathrm{Diag}\{ \left[ \mathbf{A}_{M,k} \right] _{\left( :,j \right)}\beta _{k,j} \} [ \boldsymbol{\phi }_1,\ldots,\boldsymbol{\phi }_{T_k} ] \nonumber
\\
&\,\,\,\,\,\,+[ \tilde{\mathbf{n}}_{\mathrm{m},k,j}^{\left( 1 \right)},\ldots,\tilde{\mathbf{n}}_{\mathrm{m},k,j}^{\left( T_k \right)} ] \nonumber
\\
&=\mathbf{\Lambda A}_{M}^{\mathrm{H}}\mathrm{Diag}\{ \tilde{\mathbf{h}}_{ k,j } \} \tilde{\mathbf{E}}_{\mathrm{s},k}+\tilde{\mathbf{N}}_{\mathrm{vir},k,j},
\end{align}
where $\tilde{\mathbf{E}}_{\mathrm{s},k}$ is defined in \eqref{YV}.  In the third equation, $\tilde{\mathbf{N}}_{\mathrm{vir},k,j}=[ \tilde{\mathbf{n}}_{\mathrm{m},k,j}^{\left( 1 \right)},\ldots,\tilde{\mathbf{n}}_{\mathrm{m},k,j}^{\left( T_k \right)}]$ is the corresponding noise matrix for the virtual user $\left( k,j \right)$. $\tilde{\mathbf{h}}_{k,j}$ in \eqref{Yvirkj} is the corresponding user-RIS channel for the virtual single-antenna user $(k,j)$ given by
\begin{align}
\tilde{\mathbf{h}}_{k,j}&\triangleq \left[ \mathbf{A}_{M,k} \right] _{\left( :,j \right)}\beta _{k,j}
=\beta _{k,j}\mathbf{a}_M\left( \varphi _{k,j},\theta _{k,j} \right).
\end{align}
\subsubsection{Estimation of remaining parameters of the cascaded channel for a typical user}
As introduced at the end of Section \ref{UPA-DFT}, in this part, we analyze the estimation of remaining parameters for a typical user \footnote{To ensure high estimation performance, the user closest to the RIS is chosen as the typical user since the received signal is strongest.}, i.e., user 1. 
We take the conjugate transpose of \eqref{Yvirkj} for user 1 and extract the $u$-th column of $\tilde{\mathbf{Y}}_{\mathrm{vir},1,j}^{\mathrm{H}}$ to obtain $\tilde{\mathbf{y}}_{u}^{\left( 1,j \right)}$ as
\begin{align}\label{yu}
\tilde{\mathbf{y}}_{u}^{\left( 1,j \right)}&=\tilde{\mathbf{E}}_{\mathrm{s},1}^{\mathrm{H}}\mathrm{Diag}\{\tilde{\mathbf{h}}_{1,j}^{*}\}\left[ \mathbf{A}_M\mathbf{\Lambda }^* \right] _{\left( :,u \right)}+[\tilde{\mathbf{N}}_{\mathrm{vir},1,j}^{\mathrm{H}}]_{\left( :,u \right)}\nonumber
\\
&=\tilde{\mathbf{E}}_{\mathrm{s},1}^{\mathrm{H}}\mathbf{w}_u+\tilde{\mathbf{n}}_{u}^{\left( 1,j \right)},
\end{align}
where $u$ is determined by $u=\mathop {\mathrm{arg}\max} \limits_{1\leqslant l\leqslant L}\| [ \tilde{\mathbf{Y}}_{\mathrm{vir},1,j}^{\mathrm{H}} ] _{\left( :,l \right)} \| _{2}^{2}$, $\tilde{\mathbf{n}}_{u}^{\left( 1,j \right)}=[ \tilde{\mathbf{N}}_{\mathrm{vir},1,j}^{\mathrm{H}} ] _{\left( :,u \right)}$ is the corresponding noise vector, and $\mathbf{w}_u$ in the second equation is given by $\mathbf{w}_l=\alpha _{l}^{*}\beta _{1,j}^{*}$ $\mathbf{a}_M\left( \omega _l-\varphi _{1,j},\mu _l-\theta _{1,j} \right) $ with $l=u$. The parameter estimation in \eqref{yu} can be formulated as a 1-sparse signal recovery problem under the single measurement vector (SMV) case as
\begin{align}\label{yuomp}
\tilde{\mathbf{y}}_{u}^{\left( 1,j \right)}=\tilde{\mathbf{E}}_{\mathrm{s},1}^{\mathrm{H}}\left( \mathbf{A}_1\otimes \mathbf{A}_2 \right) \boldsymbol{\tau }+\tilde{\mathbf{n}}_{u}^{\left( 1,j \right)},
\end{align}
where $\mathbf{A}_1\in \mathbb{C} ^{M_1\times F_1}$ and $\mathbf{A}_2\in \mathbb{C} ^{M_2\times F_2}$ are the overcomplete dictionary matrices. The range of atoms are $[ -\small{\frac{d}{\lambda _c},( 2-\small{\frac{4}{F_1}} ) \frac{d}{\lambda _c}}]$ and $[ -\small{\frac{d}{\lambda _c},( 2-\small{\frac{4}{F_1}} ) \frac{d}{\lambda _c}}]$ with resolutions given by $2/F_1$ and $2/F_2$, respectively. $\boldsymbol{\tau }\in \mathbb{C} ^{F_1F_2\times 1}$ is a sparse vector with one non-zero element. Utilizing SOMP of the Algorithm 1, the estimate of $\left( \omega _u-\varphi _{1,j} \right)$, $\left( \mu _u-\theta _{1,j} \right)$ and $\alpha _{u}^{*}\beta _{1,j}^{*}$ can be obtained.

With the estimated $\mathbf{w}_u$, $\left( \omega _u-\varphi _{1,j},\mu _u-\theta _{1,j} \right)$ and $\alpha _{u}^{*}\beta _{1,j}^{*}$, we use the correlation relationship between different paths of the RIS-BS channel $\mathbf{H}$ to estimate the vectors $\mathbf{w}_u$ as well as the corresponding parameters $\left( \omega _l-\varphi _{1,j},\mu _l-\theta _{1,j} \right)$ and $\alpha _{l}^{*}\beta _{1,j}^{*}$ for $l \ne u$. Specifically, the above-mentioned parameters estimation can be converted to the estimation of the correlation factors $\gamma _l \triangleq\alpha _{l}^{*}/\alpha _{u}^{*}$ and $\left( \Delta \omega _l,\Delta \mu _l \right) \triangleq \left( \omega _l-\omega _u,\mu _l-\mu _u \right)$ for $\forall l\ne u$. Now, we extract the $l$-th $\left( l\ne u \right)$ row of the matrix $\tilde{\mathbf{Y}}_{\mathrm{vir},1,j}$ to obtain  $\tilde{\mathbf{y}}_{l}^{\left( 1,j \right)}$ as
\begin{align}\label{yl1j}
\tilde{\mathbf{y}}_{l}^{\left( 1,j \right)}
&=\tilde{\mathbf{E}}_{\mathrm{s},1}^{\mathrm{H}}\mathbf{a}_M\left( \omega _l-\varphi _{1,j},\mu _l-\theta _{1,j} \right) \alpha _{l}^{*}\beta _{1,j}^{*}+\tilde{\mathbf{n}}_{l}^{\left( 1,j \right)}\nonumber
\\
&=\tilde{\mathbf{E}}_{\mathrm{s},1}^{\mathrm{H}}\mathrm{Diag}\left\{ \mathbf{w}_u \right\} \mathbf{a}_M\left( \Delta \omega _l,\Delta \mu _l \right) \gamma _l+\tilde{\mathbf{n}}_{l}^{\left( 1,j \right)}\nonumber
\\
&=\mathbf{v }\left( \Delta \omega _l,\Delta \mu _l \right) \gamma _l+\tilde{\mathbf{n}}_{l}^{\left( 1,j \right)},
\end{align}
where $\mathbf{v }\left( \Delta \omega _l,\Delta \mu _l \right) \triangleq\tilde{\mathbf{E}}_{\mathrm{s},1}^{\mathrm{H}}\mathrm{Diag}\left\{ \mathbf{w}_u \right\} \mathbf{a}_M\left( \Delta \omega _l,\Delta \mu _l \right)$.
Similar to the problem associated with \eqref{csnlls}, the SNL-LS algorithm is used $(L-1)$ times to obtain the estimate of $\left\{ \left( \Delta \omega _l,\Delta \mu _l \right) \right\} _{l=1}^{L}$ and $\left\{ \gamma _l \right\} _{l=1}^{L}$. With the obtained $\left( \omega _u-\varphi _{1,j} \right)$, $\left( \mu _u-\theta _{1,j} \right)$ and $\alpha _{u}^{*}\beta _{1,j}^{*}$, the estimate of $\left\{ \left( \omega _l-\varphi _{1,j} \right) \right\} _{l=1}^{L}$, $\left\{ \left( \mu _l-\theta _{1,j} \right) \right\} _{l=1}^{L}$ and $\left\{ \alpha _{l}^{*}\beta _{1,j}^{*} \right\} _{l=1}^{L}$ can be calculated.

Apparently, by repeating the estimation process from \eqref{yu} to \eqref{yl1j} $J_1$ times, all the parameters of the cascaded channel for virtual user $\left( 1,j \right)$ for $\forall j\in \left\{ 1,\ldots,J_1 \right\}$ can be obtained. 
However, we can adopt the following alternative method. 
Specifically, the estimation of the typical virtual user $\left( 1,1 \right)$ is first considered, and $\{ \left( \omega _l-\varphi _{1,1} \right)\} _{l=1}^{L}$, $\{ \left( \mu _l-\theta _{1,1} \right) \} _{l=1}^{L}$ and $\{ \alpha _{l}^{*}\beta _{1,1}^{*} \} _{l=1}^{L}$ are estimated. Then, the estimates of correlation factors 
	$\{ \bigtriangleup \omega _l=\omega _l-\omega _u,\bigtriangleup \mu _l=\mu _l-\mu _u,\gamma _l=\frac{\alpha _{l}^{*}}{\alpha _{u}^{*}} \} _{l=1}^{L}$ are used to obtain the estimates of parameters for the remaining virtual users $\left( 1,j \right) ,\left( j\ne 1 \right)$, i.e., $\left( \omega _u-\varphi _{1,j}\right)$, $\left(\mu _u-\theta _{1,j} \right)$ and $\alpha _{u}^{*}\beta _{1,j}^{*}$ for $ j\ne 1 $. Finally, the estimate $\hat{\mathbf{G}}_1$ is obtained based on \eqref{G}.
	The remaining parameters estimation of the cascaded channel for the typical user is summarized in \textbf{Algorithm} \ref{alg2}.

\subsubsection{Estimation of the remaining parameters of the cascaded channels for other users}
For other users, the EVSA method is still utilized. Specifically, it involves converting the estimation of remaining multi-antenna users into the estimation of $\sum_{k=2}^K{J_k}$ single-antenna users, each with only one path. Furthermore, based on $\hat{\mathbf{G}}_1$, the substitutes of the matrices $\mathbf{\Lambda}$ and $\mathbf{A}_{M}$ are constructed, denoted by $\mathbf{\Lambda }_{\mathrm{rec}}$ and $\mathbf{A}_{M,\mathrm{rec}}$, respectively, thereby reducing the required pilot overhead for the estimation of the remaining parameters of $\mathbf{G}_k$ for $\forall k\ne 1$.

Specifically, $\mathbf{\Lambda }_{\mathrm{rec}}$ and $\mathbf{A}_{M,\mathrm{rec}}$ are constructed as
\begin{align}\label{lrec}
	\mathbf{\Lambda }_{\mathrm{rec}}\triangleq\left( \alpha _{u}^{*}\beta _{1,1}^{*}\mathrm{Diag}\left\{ \gamma _1,\gamma _2,...,\gamma _L \right\} \right) ^*,
\end{align}
\begin{align}\label{AMrec}
	\mathbf{A}_{M,\mathrm{rec}}
	\triangleq \mathrm{Diag}\left\{ \mathbf{a}_M\left( \omega _u-\varphi _{1,1},\mu _u-\theta _{1,1} \right) \right\} \mathbf{A}_{\Delta}.
\end{align}
In \eqref{lrec}, $\alpha _{u}^{*}\beta _{1,1}^{*}$ and $\left\{ \gamma _l \right\} _{l=1}^{L}$ have been estimated from \eqref{yuomp}. In \eqref{AMrec}, $\left( \omega _u-\varphi _{1,1},\mu _u-\theta _{1,1} \right)$ have been estimated from \eqref{yuomp}, and $\mathbf{A}_{\Delta}\triangleq\left[ \mathbf{a}_M\left( \Delta \omega _1,\Delta \mu _1 \right) ,\ldots,\mathbf{a}_M\left( \Delta \omega _L,\Delta \mu _L \right) \right]$, where $\{ \left( \Delta \omega _l,\Delta \mu _l \right) \} _{l=1}^{L}$ are estimated from \eqref{deltwv}. Then, a substitute of $\tilde{\mathbf{h}}_{k,j}$, denoted by $\tilde{\mathbf{h}}_{\mathrm{rec},k,j}$, is determined, ensuring that the equality holds: $\mathbf{\Lambda A}_{M}^{\mathrm{H}}\mathrm{Diag}\{ \tilde{\mathbf{h}}_{k,j} \} \tilde{\mathbf{E}}_{\mathrm{s},k}=\mathbf{\Lambda }_{\mathrm{rec}}\mathbf{A}_{M,\mathrm{rec}}^{\mathrm{H}}\mathrm{Diag}\{ \tilde{\mathbf{h}}_{\mathrm{rec},k,j} \} \tilde{\mathbf{E}}_{\mathrm{s},k}$. We can derive the $\tilde{\mathbf{h}}_{\mathrm{rec},k,j}$ as
\begin{align}\label{hreckj}
\tilde{\mathbf{h}}_{\mathrm{rec},k,j}&=\mathrm{Diag}\left\{ \mathbf{a}_M\left( -\varphi _{1,1},-\theta _{1,1} \right) \right\} \tilde{\mathbf{h}}_{k,j}/\beta _{1,1}
\nonumber
\\
&=\mathbf{a}_M\left( \varphi _{k,j}-\varphi _{1,1},\theta _{k,j}-\theta _{1,1} \right) \beta _{k,j}/\beta _{1,1}.
\end{align}
Based on the reconstruction, we have $\tilde{\mathbf{Y}}_{\mathrm{rec},k,j}$ for $k \ne 1$ as
\begin{align}
\tilde{\mathbf{Y}}_{\mathrm{rec},k,j}=\mathbf{\Lambda }_{\mathrm{rec}}\mathbf{A}_{M,\mathrm{rec}}^{\mathrm{H}}\mathrm{Diag}\{ \tilde{\mathbf{h}}_{\mathrm{rec},k,j} \} \tilde{\mathbf{E}}_{\mathrm{s},k}+\tilde{\mathbf{N}}_{\mathrm{vir},k,j}.
\end{align}
Vectorizing $\tilde{\mathbf{Y}}_{\mathrm{rec},k,j}$, we can obtain measurement vector $\tilde{\mathbf{x}}_{k,j}\in \mathbb{C} ^{LT_k\times 1}$ of the virtual user $\left( k,j \right) (k \ne 1)$ given by
\begin{align}\label{xkj}
\tilde{\mathbf{x}}_{k,j}
&=\mathrm{vec}( \mathbf{\Lambda }_{\mathrm{rec}}\mathbf{A}_{M,\mathrm{rec}}^{\mathrm{H}}\mathrm{Diag}\{ \tilde{\mathbf{h}}_{\mathrm{rec},k,j} \} \tilde{\mathbf{E}}_{\mathrm{s},k}+\tilde{\mathbf{N}}_{\mathrm{vir},k,j} ) \nonumber
\\
&=( \tilde{\mathbf{E}}_{\mathrm{s},k}^{\mathrm{T}}\diamond ( \mathbf{\Lambda }_{\mathrm{rec}}\mathbf{A}_{M,\mathrm{rec}}^{\mathrm{H}} ) ) \tilde{\mathbf{h}}_{\mathrm{rec},k,j}+\tilde{\mathbf{n}}_{\mathrm{vec},k,j},
\end{align}
where $\tilde{\mathbf{n}}_{\mathrm{vec},k,j}=\mathrm{vec}( \tilde{\mathbf{N}}_{\mathrm{vir},k,j} )$ is the corresponding equivalent noise vector of the virtual user $(k,j)$.
Following the same steps as solving the estimation problem associated with \eqref{yu}, the estimation of parameters in \eqref{xkj} can be obtained by solving the 1-sparse signal recovery problem as
\begin{align}\label{vecxkj}
\tilde{\mathbf{x}}_{k,j}=( \tilde{\mathbf{E}}_{\mathrm{s},k}^{\mathrm{T}}\diamond ( \mathbf{\Lambda }_{\mathrm{rec}}\mathbf{A}_{M,\mathrm{rec}}^{\mathrm{H}} ) ) ( \mathbf{A}_1\otimes \mathbf{A}_2 ) \boldsymbol{\varrho }+\tilde{\mathbf{n}}_{\mathrm{vec},k,j},
\end{align}
where $\mathbf{A}_1$ and $\mathbf{A}_2$ are the overcomplete dictionary matrices defined in \eqref{yuomp}, and $\boldsymbol{\varrho }\in \mathbb{C} ^{F_1F_2\times 1}$ is a spare vector with one non-zero element. By utilizing SOMP algorithm, we can obtain the estimate of $\tilde{\mathbf{h}}_{\mathrm{rec},k,j},\forall j\in \{ 1,\ldots,J_k \}$, and the corresponding channel parameters, i.e., $\left\{ \varphi _{k,j}-\varphi _{1,1} \right\} _{j=1}^{J_k}$, $\{ \theta _{k,j}-\theta _{1,1} \} _{j=1}^{J_k}$ and $\{ \beta _{k,j}/\beta _{1,1}\} _{j=1}^{J_k}$. Thus, the estimate of $\mathbf{G}_k$ $(k \ne 1)$, denoted $ \hat{\mathbf{G}}_k$ $(k \ne 1)$, is obtained.
The remaining parameters estimation of the cascaded channels for other users is summarized in \textbf{Algorithm} \ref{alg3}.

\begin{algorithm}[t]
	\caption{ Estimation of Remaining Parameters of the Cascaded Channels for Other Users}
	\label{alg3}
	\renewcommand{\algorithmicrequire}{\textbf{Input:}}
	\renewcommand{\algorithmicensure}{\textbf{Output:}}
	\begin{algorithmic}[1]
		\REQUIRE The equivalent measurement matrix $\{ \tilde{\mathbf{Y}}_{\mathrm{vir},k,j} \} _{j=1}^{J_k}$ for $k\in \{ 2,\ldots,K \} $, the overcomplete dictionary matrices $\mathbf{A}_1$ and $\mathbf{A}_2$ and the estimated $\{ ( \widehat{\omega _l-\varphi _{1,1} }) \} _{l=1}^{L}$, $\{ ( \widehat{\mu _l-\theta _{1,1}} ) \} _{l=1}^{L}
		$ and $\{ \widehat{\alpha _{l}^{*}\beta _{1,1}^{*}} \} _{l=1}^{L}$. 
		\STATE Construct $\mathbf{\Lambda }_{\mathrm{rec}}$, $\mathbf{A}_{M,\mathrm{rec}}$ based on \eqref{lrec}, \eqref{AMrec}, respectively.
		\STATE Calculate the equivalent dictionary as
		\begin{align}
			\mathbf{\Sigma }_k=( 	\tilde{\mathbf{E}}_{\mathrm{s},k}^{\mathrm{T}}\diamond ( \mathbf{\Lambda }_{\mathrm{rec}}\mathbf{A}_{M,\mathrm{rec}}^{\mathrm{H}} ) ) ( \mathbf{A}_1\otimes \mathbf{A}_2 ).
		\end{align}
		\FOR{$1\leqslant j\leqslant J_k$}	
		\STATE Calculate $\tilde{\mathbf{x}}_{k,j}= \mathrm{vec}( \tilde{\mathbf{Y}}_{\mathrm{rec},k,j}) $.
		\STATE Find $p=\mathop {\mathrm{arg}\max} \limits_{i=1,...,F_1F_2}\| \left( \mathbf{\Sigma }_k \right) _{\left( :,p \right)}^{\mathrm{H}}\tilde{\mathbf{x}}_{k,j} \|$.
		\STATE Obtain $\widehat{\beta _{k,j}/\beta _{1,1}}$, and calculate the index as\\
		\quad\quad\quad\quad	$p_1=\lceil \frac{p}{F_2} \rceil,p_2=p-F_2\left( p_1-1 \right).$
		\STATE Find the estimates $\widehat{\varphi _{k,j}-\varphi _{1,1}}$, $\widehat{\theta _{k,j}-\theta _{1,1}}$ by referring the indices from dictionary matrices $\mathbf{A}_1$ and $\mathbf{A}_2$.
		\STATE Calculate the estimates as 
				$\widehat{\omega _l-\varphi _{k,j}}=\widehat{\omega _l-\varphi _{1,1}}-\widehat{\varphi _{k,j}-\varphi _{1,1}}$, 
				$\widehat{\mu _l-\theta _{k,j}}=\widehat{\mu _l-\theta _{1,1}}-\widehat{\theta _{k,j}-\theta _{1,1}}$, and $\widehat{\alpha _{l}^{*}\beta _{k,j}^{*}}=\widehat{\alpha _{l}^{*}\beta _{1,1}^{*}}\widehat{\beta _{k,j}/\beta _{1,1}}$ for $1\leqslant l\leqslant L$.
		\ENDFOR	
		\STATE Obtain $\{ \left( \omega _l-\varphi _{k,j} \right) \} _{j=1}^{J_k}$, $\{ \left( \mu _l-\theta _{k,j} \right) \} _{j=1}^{J_k}$, $\{ \alpha _{l}^{*}\beta _{k,j}^{*} \} _{j=1}^{J_k}$ for $\forall l\in \left\{ 1,\ldots,L \right\}$ and $\forall k\in \left\{ 2,\ldots,K \right\}$.
		\ENSURE $\hat{\mathbf{G}}_k,k\in \left\{ 2,...,K \right\}$.     
	\end{algorithmic}
\end{algorithm}

\section{Analysis of Pilot Overhead and Complexity}\label{pilotanalysis}
In the section, the pilot overhead and complexity of the proposed method is analyzed. We assume $\left\{ I_k \right\} _{k=1}^{K}=I$, $\left\{ J_k \right\} _{k=1}^{K}=J$, $\left\{ Q_k \right\} _{k=1}^{K}=Q$ and $\left\{ X_k \right\} _{k=1}^{K}=X$. The number of pilot signals have a significant impact on the sparse signal recovery problem.
To solve a $h$-sparse signal with dimensions $n$ recovery problem, the number of measurements $m$ is required to be on the order of $\mathcal{O} \left( h\log \left( n \right) \right)$ \cite{omppiot}. And we solve it by OMP algorithm whose complexity is $\mathcal{O} \left( mnh \right)$ \cite{complexity}. The number of measurements required for Kronecker compressed sensing also conforms to this principle \cite{KCS}.
\subsection{Pilot overhead}
\subsubsection{The pilot overhead analysis for Segment $\uppercase\expandafter{\romannumeral1}$}
In the direct channel estimation, the pilot overhead required for the estimation of the AoDs are analyzed. For the $I$-sparse recovery problem associated with \eqref{CkHVAD}, the dimension of the equivalent sensing matrix $\mathbf{F}=\mathbf{S}_{k}^{\mathrm{H}}\left( \mathbf{D}_1\otimes \mathbf{D}_2 \right)$ is $T_{\mathrm{d}}\times B_1B_2$ where $B_1\geqslant Q_{1}$, $B_2\geqslant Q_{2}$.
Therefore, the pilot overhead required should satisfy $T_{\mathrm{d}}\geqslant \mathcal{O} \left( I\log \left( B_1B_2 \right) \right) \geqslant \mathcal{O} \left( I\log \left( Q_{1}Q_{2} \right) \right) =\mathcal{O} \left( I\log \left( Q \right) \right)$, which exactly satisfies the dimensionality condition for the existence of the right pseudo-inverse of $\hat{\mathbf{D}}_{Q,k}^{\mathrm{H}}\mathbf{S}_k$ in \eqref{persudo}.
	On the other hand, in the estimation of AoDs of the user-RIS channels, following the same rationale as above, the pilot overhead required should satisfy $T_{\mathrm{AoD}}\geqslant \mathcal{O} \left( J\log \left( Q \right) \right)$ to solve the $J$-sparse recovery problem associated with \eqref{YKHVAD}.
	Therefore, for the proposed channel estimation protocol shown in Fig. \ref{fig1}, Segment $\mathrm{\uppercase\expandafter{\romannumeral1}}$ satisfies $T_{\mathrm{f}}=\max \left\{ T_{\mathrm{d}},T_{\mathrm{AoD}} \right\} =\max \log \left( Q \right)\left\{ I ,J \right\}$. 
	In conclusion, for all users, the pilot overhead required in Stage $\mathrm{\uppercase\expandafter{\romannumeral1}}$ and $\mathrm{\uppercase\expandafter{\romannumeral2}}$ satisfies $2KT_{\mathrm{f}}\geqslant \mathcal{O} \left( 2K\max \left( Q \right)\left\{ I\log  ,J \right\} \right) $.

\begin{table}[t] %
	\centering
	\begin{threeparttable}
	\caption{Pilot overhead of different estimation methods}
	\label{tab1}
	\footnotesize
	\setlength{\tabcolsep}{1pt} %
	\renewcommand{\arraystretch}{1.3} %
	\begin{tabular}{c|c} %
		\specialrule{1pt}{0pt}{0pt} %
		\rowcolor{gray!25}\textbf{Methods} & \textbf{Minimum Pilot Overhead}  \\
		\hline
		Proposed & $2K\log \left( Q \right)\max \left\{ I ,J\right\} +J\left( L+K-1 \right) \log \left( M \right) /L-KJ$  \\
		\rowcolor{gray!10}Proposed-O & $2K\log \left( Q \right)\max \left\{ I ,J\right\} +JK\log \left( M \right) -KJ$ \\
		ON-OFF & $IK\log \left( Q \right) +JK\log \left( Q \right) +JK\log \left( M \right)$  \\
		\rowcolor{gray!10}Direct-OMP& $2K\log \left( Q \right)\max \left\{ I ,J\right\} +JLK\log \left( M \right)$ \\
		SBL &  $2K\log \left( Q \right) \max \left\{ I,J \right\} +\gg JLK\log \left( M \right)$\\
		\specialrule{1pt}{0pt}{0pt}
	
	\end{tabular}
	\begin{tablenotes} 
		\item For SBL, the pilot overhead is often significantly larger \cite{sbl}.
	\end{tablenotes}
	\end{threeparttable}
\end{table}
\subsubsection{The pilot overhead analysis for Segment $\uppercase\expandafter{\romannumeral2}$}
The number of $T_k$ and $X_k$ in Fig. \ref{fig1} are analyzed.
For both the typical user and other users, to ensure the existence of right pseudo inverse for $\hat{\mathbf{A}}_{Q,k}^{\mathrm{H}}\mathbf{S}_{\mathrm{s},k}$ in \eqref{AYAS}, the equality $X_k\geqslant J_k$ should hold. The pilots required for the typical user and other users are different.
	For the estimation of the remaining parameters of the typical user, we only need to solve the 1-sparse signal recovery problem as shown in \eqref{yuomp}.
	Therefore, $T_1$ satisfies $T_1 \geqslant \mathcal{O} \left( \log \left( M \right) \right)$. Taking into account the above analyses and the designed protocol in Fig. \ref{fig1}, we conclude that the pilot overhead required to estimate the cascaded channel for the typical user $1$, i.e., the number of time slots from time block 3 to time block $X_k+2$, is $X_1\left( T_1-1 \right) \geqslant \mathcal{O} \left( J_1\left( \log \left( M \right) -1 \right) \right)$.
	Then, the pilot overhead for the other users is considered, i.e., user $k\in \left\{ 2,\ldots,K \right\}$.
	For the estimation of the remaining parameters of other users, the pilot overhead should satisfy $LT_k\geqslant \mathcal{O} \left( \log \left( M \right) \right)$, i.e., $T_k\geqslant \left( \log \left( M \right) /L \right) $ to deal with the 1-sparse recovery problem associated with \eqref{vecxkj}.
	Considering $(K-1)$ users in total, the pilot overhead required for other users is given by $\left( K-1 \right) X_k\left( T_k-1 \right) \geqslant \left( K-1 \right) \mathcal{O} \left( J\left( \log \left( M \right) /L-1 \right) \right)$.
On the other hand, as presented in Fig. \ref{fig1} and Fig. \ref{figchongzu}.
The inequality $T_{\mathrm{f},k}\geqslant X_k,\forall k\in \mathcal{K}$ should hold so that the re-combined signal $\mathbf{Y}_{\mathrm{s},k,x_k}$ in \eqref{YSKX} can be obtained for $1\leqslant x_k\leqslant X_k$. Therefore, for an MU-MIMO system, the minimum required pilot overhead is $2K\log \left( Q \right)\max \left\{ I ,J\right\} +J\left( L+K-1 \right) \log \left( M \right) /L-KJ$. Table \ref{tab1} summarizes the pilot overhead of different methods.
\subsection{Complexity Analysis}
The complexity of Stage $\mathrm{\uppercase\expandafter{\romannumeral1}}$ stems from the $I$-sparse signal recovery problem \eqref{CkHVAD} with dimensions $B_1B_2$ requiring $I\log \left( Q \right)$ measurements, which has the complexity of $\mathcal{O} \left( I^2\log \left( Q \right) B_1B_2 \right)$. And, it also generates from $I$ SNL-LS problems as \eqref{snlls} with complexity of $\mathcal{O} \left( INg_1g_2 \right)$, where $g_1$ and $g_2$ are the number of grid points in $\sigma _{k,i}$ and $\xi _{k,i}$.
In Stage $\mathrm{\uppercase\expandafter{\romannumeral2}}$, it mainly results from the sparse signal recovery problem \eqref{YKHVAD} with the complexity of $\mathcal{O} \left( J^2\log \left( Q \right) B_1B_2 \right)$.
In Stage $\mathrm{\uppercase\expandafter{\romannumeral3}}$, we use the UPA-DFT algorithm to estimate AoAs of the RIS-BS with complexity of $\mathcal{O} \left( NT_k\log \left( N \right) \right)= \mathcal{O} \left( N\log \left( M \right) \log \left( N \right) \right) $.  Through EVSA pre-processing and correlation-factors, low complexity is achieved. For user 1, the complexity of Algorithm \ref{alg2} is mainly from the $J$ sparse signal recovery problems as \eqref{yuomp} with $\mathcal{O} \left( JF_1F_2\log \left( M \right) \right)$, and from $L-1$ SNL-LS problems as \eqref{deltwv} with $\mathcal{O} \left( \left( L-1 \right) \log \left( M \right) d_1d_2 \right)$, where $d_1$ and $d_2$ are the number of grid points in $\Delta \omega _l$ and $\Delta \mu _l$. For other users $k,1\leqslant k\leqslant K$, the complexity arises from solving $J$ problems as \eqref{vecxkj}, which has the complexity of $\mathcal{O} \left( J\log \left( M \right) F_1F_2 \right)$. 
	
\section{Simulation Results}\label{sim}
\begin{figure}[b]
	\centering
	\includegraphics[width=8cm]{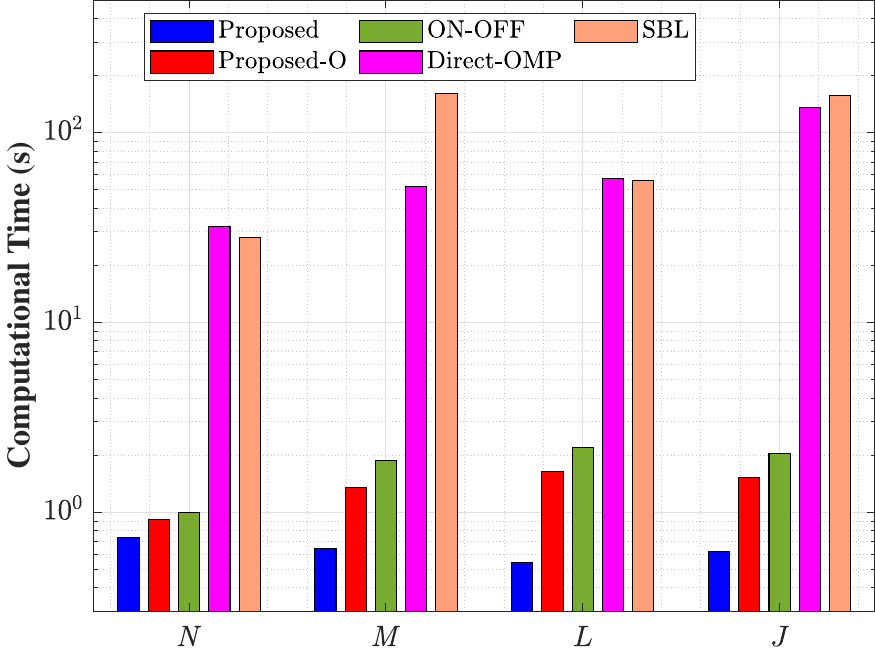}
	\caption{Computational time of the methods under different variables.}
	\label{time} 
\end{figure}

In this section, the simulation results are presented to evaluate the effectiveness of the proposed unified channel estimation protocol. The narrow-band system is considered with $30$ GHz carrier frequency and $100$ MHz transmission bandwidth. The generated channel gains $\zeta _{k,i}$, $\beta _{k,j}$ and $\alpha _l$ follow a complex Gaussian distribution, i.e. $\zeta _{k,i}\sim \mathcal{C} \mathcal{N} \left( 0,10^{-3}D_\mathrm{UB}^{-3.9} \right)$, $\beta _{k,j}\sim \mathcal{C} \mathcal{N} \left( 0,10^{-3}D_\mathrm{UR}^{-2.2} \right)$ and $\alpha _l\sim \mathcal{C} \mathcal{N} \left( 0,10^{-3}D_\mathrm{RB}^{-2.8} \right)$. $D_{\mathrm{UB}}=$ 100 m, $D_{\mathrm{UR}}=$ 40 m and $D_{\mathrm{RB}}=$ 80 m denote the distances between users and the BS, between users and the RIS and between the RIS and the BS.
The signal-to-noise ratio (SNR) is defined as SNR $=10\log \left( 10^{-6}D_{\mathrm{RB}}^{-2.8}D_{\mathrm{UR}}^{-2.2}+10^{-3}D_{\mathrm{UB}}^{-3.75} \right) /\delta ^2$.
Moreover, the parameters $N_1=N_2$, $M_1=M_2$, $K=5$, $Q_{k_1}=Q_{k_2}=6$, $\{ I_k \} _{k=1}^{K}=I$ and $\{ J_k \} _{k=1}^{K}=J$. Additionally, the number of time blocks of Segment $\mathrm{\uppercase\expandafter{\romannumeral2}}$ are set the same value, i.e., $\{ X_k \} _{k=1}^{K}=X$. And, the 1-bit quantized RIS is used, i.e. the phase shift of elements are selected from $\left\{ 0,\pi \right\}$.
In terms of the estimation performance metric, the normalized mean square error (NMSE) is chose, defined as $\mathrm{NMSE}=\mathbb{E} (\sum_{k=1}^K{|| \hat{\mathbf{G}}_{k}-\mathbf{G}_{k} || _{F}^{2}})/\mathbb{E} (\sum_{k=1}^K{|| \mathbf{G}_{k} || _{F}^{2}})$ for the cascaded channel estimation and $\mathrm{NMSE}=\mathbb{E} (\sum_{k=1}^K{|| \hat{\mathbf{H}}_{\mathrm{d},k}-\mathbf{H}_{\mathrm{d},k} || _{F}^{2}})/\mathbb{E} (\sum_{k=1}^K{\left\| \mathbf{H}_{\mathrm{d},k} \right\| _{F}^{2}})$ for the direct channel estimation. We compare the proposed framework with the ON-OFF framework and evaluate different specific estimation algorithms within the proposed framework. The detailed descriptions of each method are as follows:


\begin{itemize}
	\item{Proposed:}
	Utilize the proposed channel estimation protocol, including signal pre-processing, algorithms of the OMP class, and EVSA method.

	\item{ON-OFF:}
	Following the framework described in \eqref{ydk} and \eqref{onoffy} in Section III-C. To ensure a fair comparison, identical specific estimation algorithms provided in the Proposed method are employed.

	\item{Proposed-Ordinary (Proposed-O):}
	Use the proposed channel estimation protocol without EVSA method for other users, namely, by simply repeating the estimation procedure for the typical user $K-1$ times. It can be used as a comparison with the Proposed method to demonstrate the advantage of EVSA method.

	\item{Direct-OMP:}
	Under the proposed framework, the vectorized direct estimation of the cascaded channel from \eqref{AYAS} transforms the estimation problem into $LJ$ sparse recovery problems with a large dictionary matrix.  

	\item{Sparse Bayesian leaning (SBL):}
	Under the proposed framework, the measurement matrix $\tilde{\mathbf{Y}}_{\mathrm{m},k}^{\left( w \right)}$ in \eqref{AYAS} is recovered using SBL. It is executed with the maximum iterations 60 and the convergence threshold $10^{-4}$.
	
	\item{Oracle:}
	It is regarded as the performance upper bound of the proposed method, where the perfect angle information is known by the BS, and the channel gains are estimated by LS algorithm in Stage $\mathrm{\uppercase\expandafter{\romannumeral1}}$ and Stage $\mathrm{\uppercase\expandafter{\romannumeral3}}$.
\end{itemize}

\begin{figure}
	\centering
	\includegraphics[width=7.5cm]{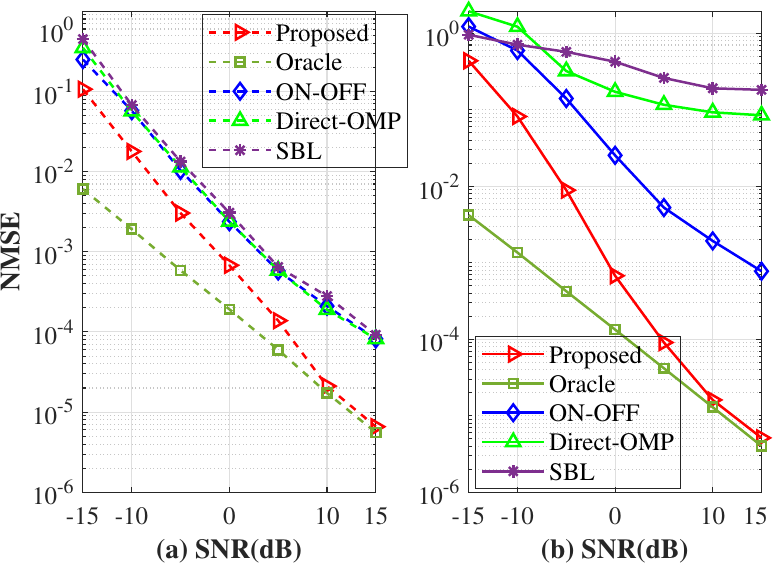}
	\caption{NMSEs versus SNR of the direct (a) and cascaded (b) channel estimation when  $Q_{k_1}(Q_{k_2})=6$, $N_1(N_2)=16$, $M_1(M_2)=16$ and $I=2$, $J=3$, $L=4$, and $T_{\mathrm{f}}=12$, $X=12$, $T=4$.}
	\label{twopara} 
\end{figure}

\begin{figure}
	\centering
	\includegraphics[width=7.5cm]{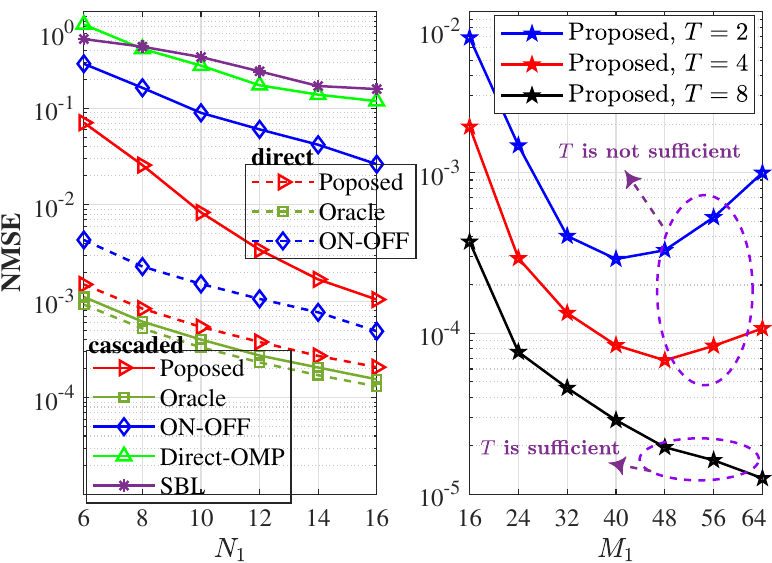}
	\caption{NMSEs versus $N_1$ and $M_1$ when SNR $= 0$ dB and $Q_{k_1}(Q_{k_2})=6$, $I=2$, $J=3$, $L=4$, and $T_{\mathrm{f}}=12$, $X=12$.}
	\label{N} 
\end{figure}

For a direct comparison of computational complexity, Fig. \ref{time} presents the computational time per simulation under varying parameters. The proposed methods exhibit substantially lower computational complexity compared to the other methods.

\begin{figure}
	\centering
	\includegraphics[width=7.5cm]{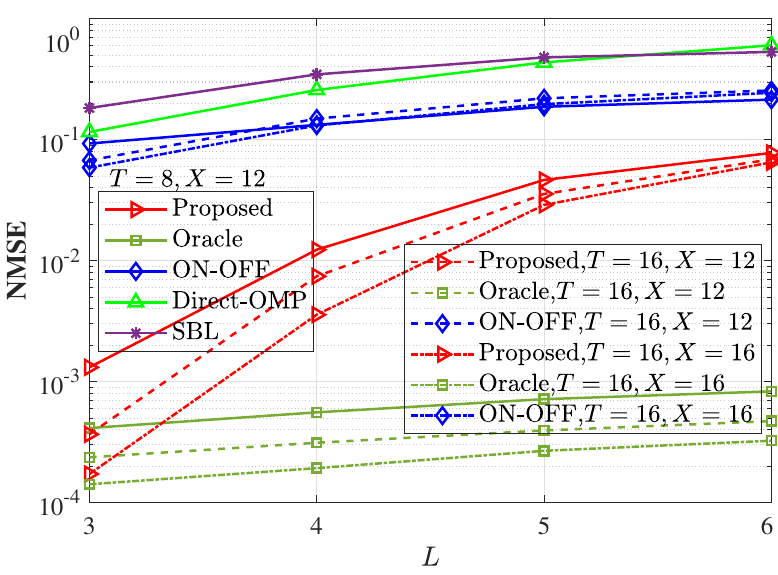}
	\caption{NMSEs versus Number of paths between the RIS and the BS for cascaded channel estimation when SNR $= 0$ dB and $Q_{k_1}(Q_{k_2})=6$, $N_1(N_2)=8$, $M_1(M_2)=16$ and $I=2$, $J=3$.}
	\label{L} 
\end{figure}

\begin{figure}[b]
	\centering
	\includegraphics[width=7.5cm]{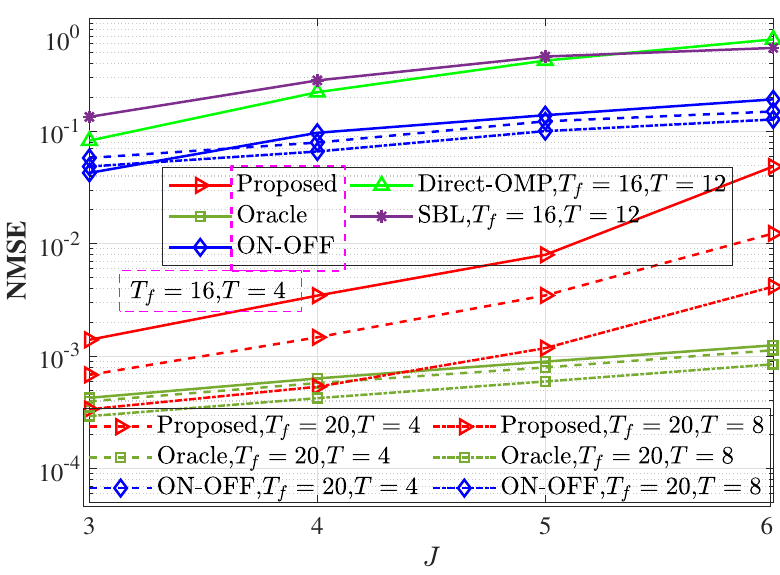}
	\caption{NMSEs vs Number of paths between the user and the RIS for cascaded channel estimation when SNR $= 0$ dB and $Q_{k_1}(Q_{k_2})=6$, $N_1(N_2)=8$, $M_1(M_2)=16$ and $I=2$, $L=4$.}
	\label{J} 
\end{figure}

Due to the flexibility of the pilot design in the proposed three-stage protocol, in simulations, the allocation of $T_{\mathrm{f},k}$ and $T_k$ obey the following rules.
First, in terms of $T_{\mathrm{f},k}$, in the proposed method, it refers to the value of time slots within the first time blocks of Segment $\mathrm{\uppercase\expandafter{\romannumeral1}}$, i.e., the time block 1 and the time block 2 as shown in Fig. \ref{fig1}. While in the ON-OFF method, $T_{\mathrm{f},k}$ time slots are allocated to estimate the direct channel, and another $T_{\mathrm{f},k}$ time slots are allocated to estimate the AoDs of user-RIS channel.
Second, since the value of $T_1$ for the typical user and $T_k$ for other users are different, we consider the average of $T_k$, denoted by $T$, seen as the average time slot per the time block. 
The same setting of $T$ also applies to the ON-OFF method.

Fig. \ref{twopara} depicts the NMSE of the direct and cascaded channels versus SNR. The gap between the proposed method and the Oracle method decreases as SNR increases for the direct and cascaded channel estimation, which implies the angles are estimated perfectly when the SNR exceeds 0 dB. 
At the framework level, the proposed method outperforms the ON-OFF method for both direct and cascaded channel estimation. This superiority stems from the fact that our method halves the noise power for the direct channel estimate and eliminates the direct-to-cascaded error propagation. At the algorithm level, under identical pilot overhead and number of paths, the proposed algorithm achieves superior estimation performance compared to both the Direct-OMP and SBL.

Fig. \ref{N} illustrates the NMSE versus the number of elements at the BS and RIS. Larger $N$ results in the increase in the number of measurements, enhancing the estimation accuracy. Furthermore, we conducted an analysis on the impact of the number of RIS elements on the proposed method. On the one hand, increasing $M$ enhances the received signal power, which is beneficial for estimation performance. However, on the other hand, a larger $M$ requires more pilot overhead, as summarized in Table \ref{tab1}. These two effects exhibit a trade-off. As shown in Fig. \ref{N}, when $M_1$ is less than 40, the first factor dominates. However, when $M_1$ becomes excessively large and pilot resources are insufficient, the second factor dominates, leading to a degradation in estimation performance. Only when the $T$ is sufficient, the estimation performance improves with an increasing $M$.


Fig. \ref{L} shows the impact of the number of paths for RIS-BS channel on the NMSE for the cascaded channel estimation. The number of parameters to be estimated for the RIS-BS channel increases, leading to precision reduction. When the number of paths for the RIS-BS channel is small and the pilot allocation meets the requirement, the performance of the proposed method is very close to that for the Oracle method. Regarding the other methods, increasing $T$ and $X$ does not result in significant performance improvement. When $L$ is large, more pilot overhead paradoxically leads to worse performance. It indicates that the error propagation impact in the ON-OFF framework severely degrades the cascaded channel estimation.
Similarly, Fig. \ref{J} depicts the impact of the number of paths of the user-RIS channel on the NMSE for the cascaded channel estimation.

\begin{figure}[t]
	\centering
	\includegraphics[width=7.5cm]{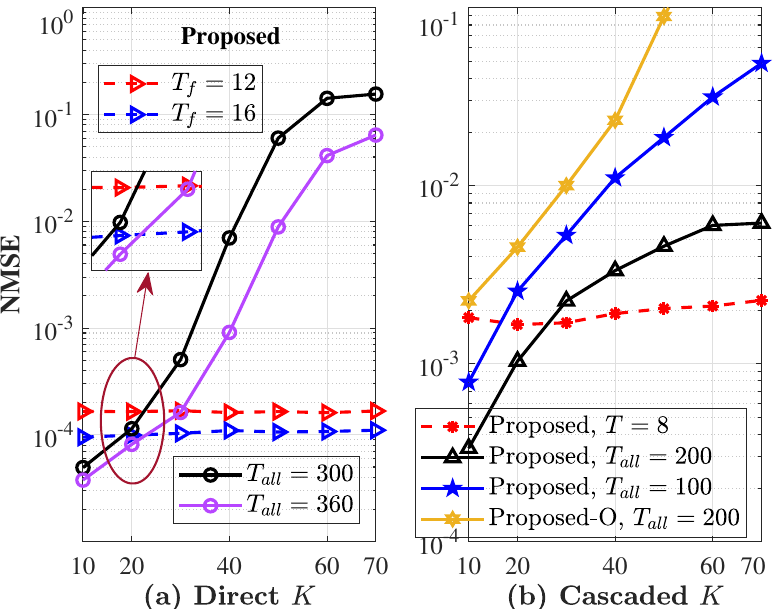}
	\caption{NMSEs vs Number of users for methods when SNR $= 0$ dB, $Q_{k_1}(Q_{k_2})=6$, $N_1(N_2)=16$, $M_1(M_2)=16$, $I=3$, $J=5$ and $L=4$.}
	\label{K} 
\end{figure}
\begin{figure}[t]
	\centering
	\includegraphics[width=7.5cm]{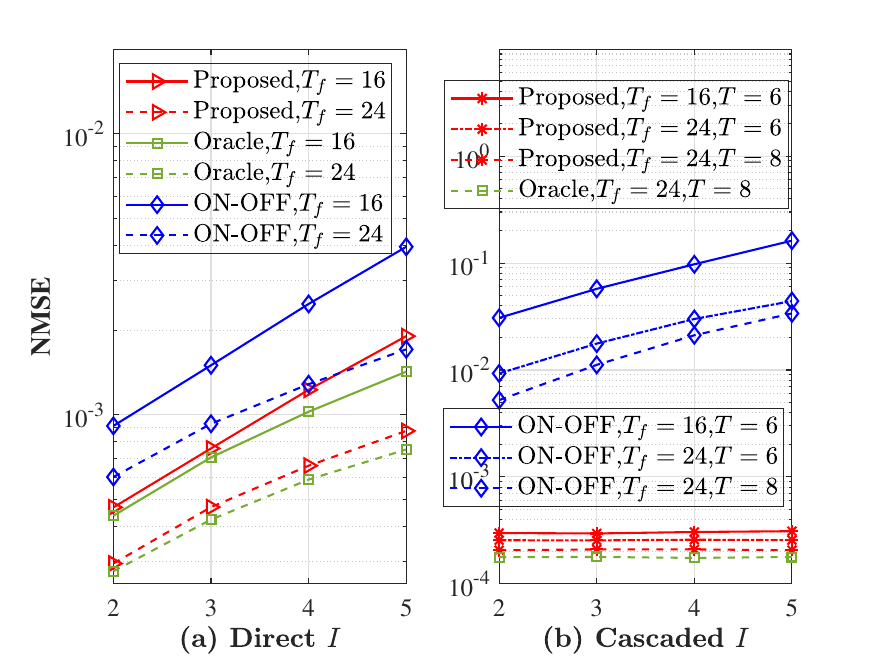}
	\caption{NMSEs vs Number of paths between the user and the BS for direct (a) and cascaded (b) channel estimation when SNR $= 0$ dB and $Q_{k_1}(Q_{k_2})=6$, $N_1(N_2)=8$, $M_1(M_2)=16$ and $J=3$, $L=4$, and $X=14$ in (b).}
	\label{I} 
\end{figure}
The performance of the proposed method in multi-user scenarios is validated in Fig. \ref{K}. We examine both cases with sufficient per-user pilots (dashed lines) and a fixed total pilot overhead (solid lines). From Fig. \ref{K} (a), it can be observed that the proposed method remains effective for the direct channel estimation even with a large number of users. In Fig. \ref{K} (b), we compare the Proposed and Proposed-O method. The results clearly demonstrate the critical role of the EVSA technique in significantly lowering pilot overhead, with its efficiency gain becoming increasingly substantial as more users are served.

Finally, we validate the effectiveness of the proposed framework in eliminating the direct-to-cascaded error propagation.
 Fig. \ref{I} shows the NMSE performances of channel estimation with the number of paths for the user-BS channel for different frameworks. From Fig. \ref{I} (a), the performance of direct channel estimation for all frameworks deteriorates as $I$ grows, due to the increasing number of parameters to be estimated. The performance of the proposed framework approaches that of the Oracle and obviously outperforms the ON-OFF framework. Fig. \ref{I} (b) reveals that the cascaded channel estimation performance of the proposed framework is not affected by $I$, whereas the performance of the ON-OFF framework degrades as $I$ increases. 
As $I$ increases, the error of direct channel estimation becomes more pronounced, degrading the cascaded channel estimation. In contrast, the proposed framework completely separates the estimation of direct and cascaded channels, which eliminates error propagation.

\section{Conclusion}\label{conl}
We investigated the channel estimation for RIS-aided MU-MIMO mmWave systems with the existence of the direct channels, which is a challenging problem that has been rarely studied due to the difficulty of independently estimating direct and cascaded channels from the received signals. 
In this paper, we proposed a novel three-stage unified channel estimation framework, overcoming the limitation of the existing ON-OFF method, which needs the RIS to be turned on-off and suffers from direct-to-cascaded error propagation. By meticulously designing the pilot vectors and the vectors of RIS, the proposed method exploits the $\pi$-differential phase shift and introduces signals recombination and orthogonal complement space of the transmitted pilot sequence to separate the direct and cascaded components from the received signals effectively. Comprehensive simulation results verified that the proposed method can achieve lower pilot overhead and higher accuracy than the existing method.

\bibliographystyle{IEEEtran}
\bibliography{IEEEabrv,reference}


\end{document}